\documentclass[a4paper]{article}

\usepackage{a4wide}
\usepackage{amsmath,amsfonts,amssymb,amsthm}

\usepackage{mathrsfs}
\usepackage{amscd,pb-diagram}

\usepackage{dutchcal}

\usepackage[active]{srcltx}

\usepackage{relsize}
\usepackage{color}

\newcounter{theorem}
\numberwithin{theorem}{section}
\numberwithin{equation}{section}
\renewcommand{\thetheorem}{\arabic{section}.\arabic{theorem}}

\newenvironment{thm}[1]{\par
\begin{sloppypar}\refstepcounter{theorem}%
\noindent{\bf #1 \thetheorem.}\it{}}{\end{sloppypar}}

\newenvironment{theorem}{\begin{thm}{Theorem}}{\end{thm}}
\newenvironment{proposition}{\begin{thm}{Proposition}}{\end{thm}}

\newenvironment{defi}[1]{\par
\begin{sloppypar}\refstepcounter{theorem}%
\noindent{\bf #1 \thetheorem.}\rm{}}{\end{sloppypar}}
\newenvironment{definition}{\begin{defi}{Definition}}{\end{defi}}

\newenvironment{remark}{\begin{defi}{Remark}}{\end{defi}}

\def\ad{\mathfrak{ad}}

\def\Op{\mathfrak{Op}} 

\def\r{\mathfrak{r}}

\def\X{\mathcal X}
\def\V{\mathcal V}

\newcommand{\Cinfpol}{$C^\infty_{\text{\sf pol}}$}
\newcommand{\varomega}{\omega\hspace*{-7pt}{\raisebox{2,5pt}{$\scriptstyle{\sim}$}}}

\def\bb1{{\rm{1}\hspace{-3pt}\mathbf{l}}}

\def\Lpol{\mathfrak{L}^2_{\text{\sf pol}}(\X)}
\def\Apol{\mathfrak{L}^1_{\text{\sf pol}}(\X)}
\def\Lbc{\mathfrak{L}^2_{\text{\sf bc}}(\X)}

\def\beq{\begin{equation}}
\def\eeq{\end{equation}}

\begin{document}

\title{Quantum observables as magnetic pseudodifferential operators}

\date{\today}

\author{Viorel Iftimie\footnote{Institute
of Mathematics Simion Stoilow of the Romanian Academy,
Bucharest.\hspace{1cm} Email: Viorel.Iftimie@imar.ro} , 
Marius M\u antoiu\footnote{Universidad de Chile, Las Palmeras 3425, Casilla 653, Santiago de Chile.\hspace{1cm} Email: mantoiu@uchile.cl} \ and Radu Purice\footnote{Institute of
Mathematics Simion Stoilow of the Romanian Academy, Bucharest.}
\footnote{Centre Francophone en Math\'{e}matique Bucarest. \hspace{1cm} Email: Radu.Purice@imar.ro}
}

\maketitle

\begin{abstract}
In a series of papers we have argued that the 'basic' physical procedure of \textit{minimal coupling} giving the quantum description of a Hamiltonian system interacting with a magnetic field, can be given a very satisfactory mathematical formulation as a \textit{twisted} Weyl quantization \cite{MP2}. In this paper we shall present a review of some of these results with some modified proofs that allow a special focus on the dependence on the behavior of the magnetic field, having in view possible developments towards problems with unbounded magnetic fields. The main new result is contained in Theorem \ref{T-ev-group} and states that the {\it the symbol of the evolution group of the self-adjoint operator defined by a real elliptic symbol of strictly positive order in a smooth bounded magnetic field is in the associated magnetic Moyal algebra, i.e. leaves invariant the space of Schwartz test functions and its dual}.
\end{abstract}


\section{Introduction}

This article is a continuation of some previous work \cite{MP1,MP3,IMP1,IMP2,IMP3}, describing magnetic pseudidifferential operators in a gauge covariant setting.

\smallskip
We consider only systems having an \textit{affine configuration space} $\X\cong\mathbb{R}^d$ with $d\geq2$. We shall denote by $\bigwedge^{\!k}\!\X$ the space of smooth $k$-forms on $\X$. \textit{The magnetic field} is then described by a \textit{closed 2-form} $B\in\bigwedge^{\!2}\!\X$, thus satisfying $dB=0$ (see lecture 13 in \cite{dA}). Due to the topological triviality of the configuration space $\X$ we can find a 1-form $A\in\bigwedge\!\X$, called \textit{a vector potential}, such that $B=dA$. Clearly the choice of $A\in\bigwedge\!\X$ is highly non-unique, different choices being related by \textit{gauge transformations} $A\mapsto A^\prime=A+d f$ for some $f\in C^2(\X;\mathbb{R})$. A usual choice of the vector potential is \textit{the transversal gauge}:
\beq\label{F-transv-gauge}
A_j(x)=-\sum_{k=1}^d\int_0^1\!ds\,B_{jk}(sx)sx_k,
\eeq
verifying $x\cdot A(x)=0$.

\smallskip
\textit{A Hamiltonian system} is described by a \textit{smooth function} $h:\Xi\rightarrow\mathbb{R}$, where $\Xi:=\X\times\X^*$ is \textit{the phase space} of the system, with $\X^*$ the dual of $\X$ as a finite dimensional real vector space, with the duality map ${<\cdot,\cdot>:\X^*\times\X\rightarrow\mathbb{R}}$ (see \cite{MP3} and the references therein).

\subsection{Notations.}

We shall use the notation $[t]:=\max\big\{k\in\mathbb{Z}\mid\,k\leq t\big\}\in\mathbb{Z}$ for the \textit{integer part} of $t\in\mathbb{R}$.

\smallskip
For any Euclidian space $\mathcal{V}\cong\mathbb{R}^N$ we denote by $\mathscr{S}(\mathcal{V})$ the Fr\'{e}chet space of Schwartz test functions and by $\mathscr{S}^\prime(\mathcal{V})$ its dual, the space of tempered distributions on $\mathcal{V}$. We shall denote by $C^\infty(\mathcal{V})$ the space of smooth functions on
$\mathcal{V}$ and by $C^\infty_{\text{\sf pol}}(\mathcal{V})$ (resp. $C^\infty_{\text{\sf pol,u}}(\mathcal{V})$) and by $BC^\infty(\V)$ its subspaces of smooth functions
that are polynomially bounded together with all their derivatives (resp. those with uniform polynomial growth on all the derivatives) or smooth and bounded
together with all their derivatives. We use the notation $\left<v\right>:=\sqrt{1+|v|^2}$ for any $v\in\mathcal{V}$.

\smallskip
When working in a Hilbert space $L^2(\mathcal{V})$ over an Euclidian space $\mathcal{V}\cong\mathbb{R}^N$ with the Lebesgue measure, we shall denote by $F(Q)$ the operator of multiplication with the measurable function $F:\mathcal{V}\rightarrow\mathbb{C}$, i.e.
$$
\big(F(Q)f\big)(v):=F(v)f(v),\ \forall v\in\mathcal{V},\quad\forall f\in L^2(\mathcal{V}).
$$
Moreover we shall denote by $\mathbb{B}\big(\mathcal{H}\big)$ and $\mathcal{U}\big(\mathcal{H}\big)$ the algebra of bounded linear operators, respectively the group of unitary linear operators on the Hilbert space $\mathcal{H}$.

\smallskip
For $k$-forms on $\X$, with $k\in\mathbb{N}$, we shall consider the spaces
$$
\mathfrak{L}^k_{\text{\sf pol}}(\X)\ :=\,\Big\{F\in\bigwedge^k\X\,\mid\,F_{j_1,\ldots,j_k}\in C^\infty_{\text{\sf pol}}(\X),\,\forall(j_1,\ldots,j_k)\in\{1,\ldots,d\}^k\Big\},
$$
$$
\mathfrak{L}^k_{\text{\sf bc}}(\X)\ :=\,\Big\{F\in\bigwedge^k\X\,\mid\,F_{j_1,\ldots,j_k}\in BC^\infty(\X),\,\forall(j_1,\ldots,j_k)\in\{1,\ldots,d\}^k\Big\}.
$$
Clearly $\mathfrak{L}^0_{\text{\sf pol}}(\X)=C^\infty_{\text{\sf pol}}(\X)$ and $\mathfrak{L}^0_{\text{\sf bc}}(\X)=BC^\infty(\X)$.

\smallskip
We shall need several types of semi-norms and weights on these spaces. We shall call \textit{weight} a positive function, verifying the properties of a semi-norm but being allowed to take also the value $+\infty$. 

\smallskip
On $\mathfrak{L}^k_{\text{\sf bc}}(\X)$ we shall use the following two families of semi-norms, indexed by $m\in\mathbb{N}$:
$$
\mu_m(F)\ :=\ \underset{(j_1,\ldots,j_k)\in\{1,\ldots,d\}^k}{\max}\Big(\,\underset{|\alpha| = m}{\max}\,\underset{x\in\X}{\sup}\left|\big(\partial^\alpha F_{j_1,\ldots,j_k}\big)(x)\right|\Big),\qquad\rho_m(F):=\underset{n\leq m}{\max}\,\mu_n(F).
$$

On $\mathfrak{L}^k_{\text{\sf pol}}(\X)$ we consider the following family of \textit{weight functions} indexed by $(p,m)\in\mathbb{R}\times\mathbb{N}$:
\beq\label{FD-seminorms-Cinfpol}
\nu^{p}_{m}(F):\ =\ \underset{(j_1,\ldots,j_k)\in\{1,\ldots,d\}^k}{\max}\Big(\,\underset{x\in\mathcal{X}}{\sup}\left<x\right>^{-p}\!\underset{|\alpha|= m}{\max}\left|\big(\partial^\alpha F_{j_1,\ldots,j_k}\big)(x)\right|\Big).
\eeq

In developing the magnetic pseudodifferential calculus an important role will be played by a specific imaginary exponential of the magnetic flux through some triangles (see \eqref{FD-omega-B} and the results in Appendix 1). 
The formulas in the Appendix 1 show the interest in defining the following family of functionals that characterize the growth of the magnetic field and its derivatives. For $B\in\Lbc$ and for any $M\in\mathbb{N}$ we define:
\beq\label{FD-weights-B}
\mathring{\mathfrak{w}}_M(B)\ :=\ \max\Big\{\mu_0(B),\underset{1\leq m\leq M}{\max}\ \underset{1\leq l\leq m}{\max}\ \underset{p_1+\cdots+p_l=m}{\max}\ \prod\limits_{s=1}^{l}\big(\mu_{p_s}(B)+\mu_{p_s-1}(B)\big)\Big\}.
\eeq

For the points of $\Xi=\X\times\X^*$ we shall use notations of the form $X=(x,\xi), Y=(y,\eta), Z=(z,\zeta)$. We recall that on $\Xi$ we have a \textit{canonical symplectic form}:
\beq\label{Xi-simpl-form}
\sigma(Y,Z)\ :=\ <\eta,z>-<\zeta,y>.
\eeq

For the space $C^\infty_{\text{\sf pol}}(\Xi)$ we shall use a family of \textit{weight functions} of the form \eqref{FD-seminorms-Cinfpol} but with four indices $(p_1,p_2,m_1,m_2)\in\mathbb{R}^2\times\mathbb{N}^2$:
\beq\label{FD-seminorms-Cinfpol-Xi}
\nu^{p_1,p_2}_{m_1,m_2}(F):=\underset{(x,\xi)\in\Xi}{\sup}\left<x\right>^{-p_1}\!\left<\xi\right>^{-p_2}\!\underset{|a|=m_1}{\max}\,\underset{|\alpha|=m_2}{\max}\,\big|\big(\partial_x^a\partial_\xi^\alpha F\big)(x,\xi)\big|,\quad\nu^p_{m_1,m_2}(F)\equiv\nu^{0,p}_{m_1,m_2}(F).
\eeq

Let us recall now a family of H\"{o}rmander type symbols (\cite{H-I}) that will play a very important role in our analysis.

\smallskip
\begin{definition}\label{D-H-symb}
For any $p\in\mathbb{R}$ we denote by $S^p(\Xi,\X)$ the following complex linear space
\beq\label{FD-H-symb}
S^p(\Xi,\X)\,:=\,\left\{F\in \text{\Cinfpol}(\Xi)\,\mid\, \underset{(x,\xi)\in\Xi}{\sup}\left<\xi\right>^{-p+|\beta|}\big|\big(\partial^\alpha_x\partial^\beta_\xi F\big)(x,\xi)\big|<\infty,\,\forall\,\alpha,\beta\in\mathbb N^d\right\},
\eeq
with the Fr\'{e}chet topology defined by the countable family of semi-norms $\{\nu^{p-m_2}_{m_1,m_2}\}$ indexed by $(m_1,m_2)\in\mathbb{N}^2$.

We also set
$$
S^\infty(\Xi,\X):=\underset{p\in\mathbb{R}}{\bigcup}S^p(\Xi,\X),\qquad S^{-\infty}(\Xi,\X):=\underset{p\in\mathbb{R}}{\bigcap}S^p(\Xi,\X),\qquad S^{-}(\Xi,\X):=\underset{p<0}{\bigcap}S^p(\Xi,\X).
$$
\end{definition}

Let us point out that \eqref{FD-H-symb} is the class $S^p_{1,0}(\Xi)$ with the notations from \cite{H-I,IMP1}. Noticing that $\nu^{p}_{m_1,m_2}(F)\leq\nu^{p-m_2}_{m_1,m_2}(F)$ for any $F\in S^p(\Xi,\X)$ and any $(m_1,m_2)\in\mathbb{N}^2$, we shall also use the following semi-norms on $S^p(\Xi,\X)$ (indexed by $(m_1,m_2)\in\mathbb{N}^2$):
$$
\rho^p_{m_1;m_0,m_2}(F):=\underset{n_1\leq m_1}{\max}\,\underset{m_0\leq n_2\leq m_0+m_2}{\max}\nu^{p}_{n_1,n_2}(F).
$$

\begin{definition}\label{D-ell-symb}
We say that a symbol $F\in S^p(\Xi,\X)$ is \textit{elliptic} if there exists two constants $(R,C)\in\mathbb{R}_+\times\mathbb{R}_+$ such that
$$
|F(x,\xi)|\geq C<\xi>^p,\ \forall (x,\xi)\in\X\times\{\xi\in\X^*\mid|\xi|\geq R\}.
$$
We denote by $S^p(\Xi,\X)_{\text{\sf ell}}$ the family of elliptic symbols of type $p\in\mathbb{R}$.
\end{definition}

\smallskip
For any $s\in\mathbb{R}$ we shall use the notation $\mathfrak{p}_s(x,\xi):=\left<\xi\right>^{s}$, defining an elliptic symbol of order $s\in\mathbb{R}$ and $\mathfrak{q}_s(x,\xi):=\left<x\right>^{s}$ (that for $s>0$ is not a H\"{o}rmander type symbol).

\smallskip
We shall use the following Fourier transforms:
\begin{align}
\mathcal{F},\mathcal{F}^-:\mathscr{S}(\X)\rightarrow\mathscr{S}(\X^*):&\nonumber\\
&\big(\mathcal{F}\phi\big)(\xi):=(2\pi)^{-d/2}\int_{\X}e^{-i<\xi,x>}\phi(x)\,dx,\nonumber\\
&\big(\mathcal{F}^-\phi\big)(\xi):=(2\pi)^{-d/2}\int_{\X}e^{i<\xi,x>}\phi(x)\,dx\nonumber,
\end{align}
\begin{align}
\mathcal{F}_*,\mathcal{F}_*^-:\mathscr{S}(\X^*)\rightarrow\mathscr{S}(\X):&\nonumber\\
&\big(\mathcal{F}_*\psi\big)(x):=(2\pi)^{-d/2}\int_{\X^*}e^{-i<\xi,x>}\psi(\xi)\,d\xi,\nonumber\\
&\big(\mathcal{F}_*^-\psi\big)(x):=(2\pi)^{-d/2}\int_{\X^*}e^{i<\xi,x>}\psi(\xi)\,d\xi,\nonumber
\end{align}
\beq
\widetilde{\mathcal{F}}:\mathscr{S}(\Xi)\rightarrow\mathscr{S}(\Xi),\quad
\big(\widetilde{\mathcal{F}}F\big)(x,\xi):=(2\pi)^{-d}\int_{\Xi}e^{i<\xi,y>-<\eta,x>}F(y,\eta)\,dy\,d\eta.\nonumber
\eeq

\subsection{The magnetic quantization.}

Given a magnetic field $B\in\bigwedge^{\!2}\!\X$ and an associated vector potential $A\in\bigwedge\!\X$, let us consider the following \textit{invariant integrals}, whose significance in constructing a \textit{gauge covariant functional calculus} has been noticed in \cite{CN,N-02}:
$$
\Gamma^A(x,y):=\int_{[x,y]}\hspace*{-10pt}A\ \ \ ,\qquad\Phi^B(x,y,z):=\int_{<x,y,z>}\hspace*{-15pt}B\ \,,
$$
where $[x,y]$ is the oriented line segment from $x\in\X$ to $y\in\X$ and $<x,y,z>$ is the oriented triangle in $\X$ having the vertices $\{x,y,z\}\subset\X$. They verify the following relations due to the Stoke' formula and the condition $dB=0$:
$$
\forall\{x,y,z\}\subset\X:\qquad\Gamma^A(x,y)+\Gamma^A(y,z)+\Gamma^A(z,x)=
\Phi^B(x,y,z),
$$
$$
\forall\{x,y,z,w\}\subset\X:\qquad\Phi^B(x,y,w)+\Phi^B(y,z,w)+\Phi^B(z,x,w)=\Phi^B(x,y,z).
$$
We shall also use the notations:
$$
\Lambda^A(x,y):=e^{-i\Gamma^A(x,y)},\qquad\Omega^B(x,y,z):=e^{-i\Phi^B(x,y,z)}.
$$
\begin{definition}\label{D-magn-W-syst}
Given a magnetic field $B\in\bigwedge^{\!2}\!\X$ and an associated vector potential $A\in\bigwedge\!\X$, we call \emph{the Magnetic Weyl system} on $\Xi=\X\times\X^*$ associated to $A\in\bigwedge\X$ the application
\beq\label{F-magn-W-syst}
W^A:\Xi\rightarrow\mathcal{U}\big(L^2(\X)\big),\qquad W^A(x,\xi)f:=\Lambda^A(Q,Q+x)W(x,\xi)f,
\eeq
where $W:\Xi\rightarrow\mathcal{U}\big(L^2(\X)\big)$ is the usual Weyl system on $\Xi=\X\times\X^*$:
\beq\label{F-W-syst}
\big(W(x,\xi)f\big)(z):=e^{-(i/2)<\xi,x>}e^{-i<\xi,z>}f(z+x),\ \forall z\in\X,\quad\forall f\in L^2(\X).
\eeq
We shall sometimes use the notation
\beq\label{W-syst-2}
U(x):=W(-x,0),\quad V(\xi):=W(0,\xi),\quad U^A(x)=W^A(-x,0).
\eeq
\end{definition}

\begin{definition}\label{D-magn-quant}
Given a magnetic field $B\in\bigwedge^{\!2}\!\X$ and an associated vector potential $A\in\bigwedge\!\X$, we define \textit{the magnetic quantization} as the application
\beq\label{F-magn-quant}
\Op^A:\mathscr{S}(\Xi)\rightarrow\mathbb{B}\big(L^2(\X)\big),\qquad\Op^A(F):=(2\pi)^{-d}\int_{\Xi}\big(\widetilde{\mathcal{F}}F\big)(x,\xi)W^A(x,\xi)\,dx\,d\xi,
\eeq
with the integral defined in the weak operator sense.
\end{definition}

\smallskip
For $A=0$ we obtain the usual Weyl quantization, that we shall denote by $\Op\equiv\Op^0$.

\smallskip
\begin{proposition}\label{P-gauge-cov} \textit{(Proposition 3.4 in \cite{MP1})}\\
Given two gauge equivalent vector potentials $A^\prime=A+df$, the corresponding magnetic quantizations are unitarily equivalent; more precisely we have
$$
\Op^{A^\prime}\!(F)\ =\ e^{if(Q)}\Op^A(F)e^{-if(Q)},\qquad\forall F\in\mathscr{S}^\prime(\Xi).
$$
\end{proposition}

\begin{proposition}\label{Gen-Diamagn-ineq} \textit{\textbf{A Diamagnetic Inequality for symbols}}\\
Suppose given a magnetic field $B\in\Lpol$. Then, for any distribution $F\in\mathscr{S}^\prime(\X^*)$ (considered as the subspace of $\mathscr{S}^\prime(\Xi)$ of distributions constant along the directions in $\X$) such that $\mathcal{F}_*F$ is a non-negative distribution (takes non-negative values on non-negative test functions) we have the inequality
\beq\label{G-D-ineq}
\left|\Op^A(F)\phi\right|\ \leq\ \Op(F)|\phi|,\quad\forall \phi\in\mathscr{S}(\X),
\eeq
where $\Op(F)$ is the usual Weyl quantization of $F\in\mathscr{S}^\prime(\X^*)\subset\mathscr{S}^\prime(\Xi)$ (obtained from \eqref{D-magn-quant} for $A=0$).
\end{proposition}

\begin{proof}
By hypothesis and using Theorem I.4.V in \cite{Schw} we conclude that $\mathcal{F}_*F$ is a temperate positive measure (see \cite{Schw}) $\mu_F$ on $\X$. Thus, for any test function $\phi\in\mathscr{S}(\X)$ we can write
$$
\Op^A(F)\phi=(2\pi)^{-d}\int_{\Xi}dx\,d\xi\,\big((\mathcal{F}^-\otimes\mathcal{F}_*)F\big)(\xi,x)\big(W^A(x,\xi)\phi\big)=(2\pi)^{-d/2}\int_{\X}\mu_F(dx)\big(U^A(x)\phi\big).
$$
Choosing one more test function $\psi\in\mathscr{S}(\X)$ we compute
$$
\big[\Op^A(F)\phi\big](\overline{\psi})=(2\pi)^{-d/2}\!\int_{\X}\mu_F(dx)\big\langle\psi,U^A(x)\phi\big\rangle_
{L^2(\X)}=
$$
$$
=(2\pi)^{-d/2}\int_{\X}\big(\left<x\right>^{-N}\!\mu_F(dx)\big)\left<x\right>^N\!\big\langle\psi,U^A(-x)\phi\big\rangle_{L^2(\X)}\,,
$$
where for $N\in\mathbb{N}$ large enough the measure $\mu_{F,N}(dx):=\left<x\right>^{-N}\mu_F(dx)$ is a finite positive measure with total mass $M_{F,N}<\infty$. On the other hand we notice that
$$
\big\langle\psi,U^A(x)\phi\big\rangle_{L^2(\X)}=\int_\X dz\,\Lambda^A(z,z+x)\overline{\psi(z)}\phi(z+x)
$$
and deduce that, for any $N\in\mathbb{N}$, there exists some $C_N>0$ such that
$$
\left|\big\langle\psi,U^A(x)\phi\big\rangle_{L^2(\X)}\right|\,\leq\,C_N\left<x\right>^{-N}\left(\underset{z\in\X}{\sup}\left<z\right>^{n_1}|\psi(z)|\right)\left(\underset{z\in\X}{\sup}\left<z\right>^{n_2}|\phi(z)|\right).
$$
Using the results in \cite{Schw} we conclude that $\Op^A(F)\phi$ is a tempered complex measure on $\X$. Thus its absolute value is a well defined positive tempered measure and we can write
$$
\left|\Op^A(F)\phi\right|\leq(2\pi)^{-d/2}\!\int_{\X}\mu_F(dx)\,\left|U^A(-x)\phi\right|=(2\pi)^{-d/2}\!\int_{\X}\mu_F(dx)\,U(-x)|\phi|=\Op(F)|\phi|.
$$
\end{proof}

\begin{description}
\item[The integral kernels.]
If we use formulas \eqref{F-magn-quant} and \eqref{F-magn-W-syst} we obtain for any $F\in\mathscr{S}(\Xi)$ and $\phi\in\mathscr{S}(\X)$ 
$$
\big(\Op^A(F)\phi\big)(z)=(2\pi)^{-d}\int_{\Xi}\big(\widetilde{\mathcal{F}}F\big)(x,\xi)\big(W^A(x,\xi)F\big)(z)\,dx\,d\xi\,=
$$
$$
=(2\pi)^{-d}\int_{\Xi}\big((\mathscr{F}^-\otimes\mathscr{F}_*)F\big)(\xi,x)\Lambda^A(z,z+x)
\left(e^{-(i/2)<\xi,x>}e^{-i<\xi,z>}f(z+x)\right)\,dx\,d\xi\,=
$$
$$
=(2\pi)^{-d/2}\int_{\X}\Lambda^A(z,z+x)\big((\bb1\otimes\mathscr{F}_*)F\big)(z+(x/2),x)f(z+x)\,dx\,=
$$
$$
=(2\pi)^{-d/2}\int_{\X}\Lambda^A(z,y)\big((\bb1\otimes\mathscr{F}_*)F\big)\big((z+y)/2,y-z\big)f(y)\,dy\,\equiv\int_{\X}\mathfrak{K}^A_F(z,y)f(y)\,dy.
$$
Thus the integral kernel of the operator $\Op^A(F)$ is
$$
\mathfrak{K}^A_F(x,y)=(2\pi)^{-d/2}\Lambda^A(x,y)\big((\bb1\otimes\mathscr{F}_*)F\big)\big((x+y)/2,y-x\big).
$$
Setting $\big(\Upsilon F\big)(x,y):=(2\pi)^{-d/2}F\big((x+y)/2,y-x\big)$, we can write 
\beq\label{F-int-kernel}
\mathfrak{K}^A_F\ =\ \Lambda^A\Upsilon(\bb1\otimes\mathscr{F}_*)F\equiv\Lambda^A\mathfrak{K}_F
\eeq
and notice that it defines isomorphisms $\mathscr{S}(\Xi)\overset{\sim}{\rightarrow}\mathscr{S}(\X\times\X)$ and 
$\mathscr{S}^\prime(\Xi)\overset{\sim}{\rightarrow}\mathscr{S}^\prime(\X\times\X)$.

\begin{remark}
The map $\Upsilon:\X\times\X\rightarrow\X\times\X$ is a linear bijection with Jacobian 1 and its inverse has the explicit action $\Upsilon^{-1}(u,v)=\big(u-v/2,u+v/2\big)$.
\end{remark}
\end{description}

The following two statements are proved in \cite{MP1}. The first one is an easy consequence of the results in Sections 50 and 51 in \cite{T}, taking into account that $\mathscr{S}(\X)$ and $\mathscr{S}^\prime(\X)$ are \textit{nuclear spaces} (\cite{T}).

\smallskip
\begin{proposition}\label{P-izo-S}
If the vector potential $A$ is in $\Apol$, then the application $\Op^A$ defines a linear and topological isomorphism 
$$
\Op^A:\mathscr{S}(\Xi)\overset{\sim}{\rightarrow}\mathcal{L}\big(\mathscr{S}^\prime(\X);\mathscr{S}(\X)\big).
$$
Here $\mathcal{L}\big(\mathscr{S}^\prime(\X);\mathscr{S}(\X)\big)$ is the space of linear continuous maps from $\mathscr{S}^\prime(\X)$ to $\mathscr{S}(\X)$ with the topology of uniform convergence on bounded subsets.
\end{proposition}

\smallskip
Using the \textit{Kernel Theorem} of L. Schwartz and \eqref{F-int-kernel} (\cite{T}), we can extend the magnetic quantization to the space of tempered distributions.

\smallskip
\begin{proposition}\label{P-izo-Sprime}
If the vector potential $A\in\Apol$, then the application $\Op^A$ defines a linear and topological isomorphism 
\beq\label{F-ext-Op}
\Op^A:\mathscr{S}^\prime(\Xi)\overset{\sim}{\rightarrow}\mathcal{L}\big(\mathscr{S}(\X);\mathscr{S}^\prime(\X)\big).
\eeq
\end{proposition}

Using Proposition \ref{P-izo-S} we notice that given a magnetic field $B\in\Lpol$, formula \eqref{F-transv-gauge} allows us to fix an associated vector potential $A\in\Apol$, and thus, for any pair of test functions 
$(\phi,\psi)\in\mathscr{S}(\Xi)\times\mathscr{S}(\Xi)$ the product $\Op^A(\phi)\Op^A(\psi)$ belongs to $\mathcal{L}\big(\mathscr{S}^\prime(\X);\mathscr{S}(\X)\big)$ and there exists a unique test function $\rho^B(\phi,\psi)\in\mathscr{S}(\Xi)$ such that
\beq\label{FD-sharp-B}
\Op^A(\phi)\Op^A(\psi)\ =\ \Op^A\big(\rho^B(\phi,\psi)\big).
\eeq

\begin{definition}\label{D-sharp-B}
For any magnetic field $B\in\Lpol$ we define the following composition map:
$$
\mathscr{S}(\Xi)\times\mathscr{S}(\Xi)\ni(\phi,\psi)\mapsto\phi\sharp^B\psi:=\rho^B(\phi,\psi)\in\mathscr{S}(\Xi),
$$
with $\rho^B(\phi,\psi)$ satisfying \eqref{FD-sharp-B}. We call it \textit{the magnetic Moyal product}. 
\end{definition}

\smallskip
Clearly $\rho^B(\phi,\psi)\in\mathscr{S}(\Xi)$ depends linearly and continuously (due also to Proposition \ref{P-izo-S}) on both variables $(\phi,\psi)\in\mathscr{S}(\Xi)\times\mathscr{S}(\Xi)$.
A straightforward computation allows us to prove that
\beq\label{F-sharp-B-1}
\big(\phi\sharp^B\psi\big)(X)=\pi^{-2d}\int_{\Xi}dY\int_{\Xi}dZ\,e^{-2i\sigma(Y,Z)}\varomega^B_x(y,z)\phi(X-Y)\psi(X-Z)=
\eeq
\beq\label{F-sharp-B-2}
=\pi^{-2d}\int_{\Xi}dY\int_{\Xi}dZ\,e^{-2i\sigma(X-Y,X-Z)}\omega^B(x,y,z)\phi(Y)\psi(Z),
\eeq
where 
\beq\label{FD-omega-B}
\varomega^B_x(y,z):=\exp\left\{{-i\int_{\mathcal{T}_x(y,z)}\hspace*{-14pt}B}\hspace*{8pt}\right\},\qquad
\omega^B(x,y,z):=\exp\left\{{-i\int_{\mathcal{T}(x,y,z)}\hspace*{-14pt}B}\hspace*{8pt}\right\},
\eeq
with $\mathcal{T}_x(y,z)$ the oriented triangle in $\X$ with vertices $x-y-z,x+y-z,x-y+z$ and $\mathcal{T}(x,y,z)$ the oriented triangle in $\X$ with vertices $y+z-x,z+x-y,x+y-z$. In the first Appendix to this paper we prove a number of estimations on the function $\varomega^B\in C^\infty_{\text{\sf pol}}\big(\X;C^\infty_{\text{\sf pol}}(\X\times\X)\big)$ that will be needed. 

\smallskip
In \cite{MP1} (Lemma 4.14 and Corollary 4.15) the following statement is proved.

\smallskip
\begin{proposition}\label{P-sharp-B}
Given a magnetic field $B\in\Lpol$,
\begin{enumerate}
\item for any pair of test functions 
$(\phi,\psi)\in\mathscr{S}(\Xi)\times\mathscr{S}(\Xi)$ the following equality holds:
\beq\label{FP-sharp-B-2}
\int_{\Xi}dX\,\big(\phi\sharp^B\psi\big)(X)\ =\ \int_{\Xi}dX\,\phi(X)\psi(X);
\eeq
\item for any three test functions 
$(\phi,\psi,\chi)\in\mathscr{S}(\Xi)\times\mathscr{S}(\Xi)\times\mathscr{S}(\Xi)$ the following equality holds:
\beq\label{FP-sharp-B-3}
\int_{\Xi}dX\,\big(\phi\sharp^B\psi\big)(X)\chi(X)\ =\ \int_{\Xi}dX\,\phi(X)\big(\psi\sharp^B\chi\big)(X)\ =\ \int_{\Xi}dX\,\psi(X)\big(\chi\sharp^B\phi\big)(X)).
\eeq
\end{enumerate}
\end{proposition}

The above result allows us to extend the magnetic Moyal product by duality and define two bilinear bicontinuous maps
\beq\label{Ext-sharp-B}
\sharp^B:\mathscr{S}^\prime(\Xi)\times\mathscr{S}(\Xi)\rightarrow\mathscr{S}^\prime(\Xi),\quad
\sharp^B:\mathscr{S}(\Xi)\times\mathscr{S}^\prime(\Xi)\rightarrow\mathscr{S}^\prime(\Xi).
\eeq

\subsection{The magnetic Moyal algebra.}

We shall briefly recall some definitions and results from \cite{MP1}. Let us define 
\beq\label{FD-Moyal-alg}
\mathfrak{M}^B(\Xi)\ :=\ \left\{F\in\mathscr{S}^\prime(\Xi)\,\mid\,F\sharp^B\phi\in\mathscr{S}(\Xi),\  \phi\sharp^BF\in\mathscr{S}(\Xi),\ \forall\phi\in\mathscr{S}(\Xi)\right\}.
\eeq
\begin{remark}\label{Moyal-alg} \textit{(Proposition 4.20 in \cite{MP1})}\\
$\mathfrak{M}^B(\Xi)$ is a unital *-algebra (with the *-involution given by the complex conjugation) containing $\mathscr{S}(\Xi)$ as a two-sided $^*$-ideal.
\end{remark}

\smallskip
\begin{definition}\label{D-magn-Moyal-alg}
We call $\mathfrak{M}^B(\Xi)$ \emph{the magnetic Moyal algebra} associated to the magnetic field $B\in\Lpol$.
\end{definition}

\smallskip
Using Proposition \ref{P-sharp-B} we can extend $\sharp^B:\mathfrak{M}^B(\Xi)\times\mathfrak{M}^B(\Xi)\rightarrow
\mathfrak{M}^B(\Xi)$ by duality to the following applications:
\beq\label{Ext-sharp-B-2}
\sharp^B:\mathfrak{M}^B(\Xi)\times\mathscr{S}^\prime(\Xi)\rightarrow
\mathscr{S}^\prime(\Xi),\quad
\sharp^B:\mathscr{S}^\prime(\Xi)\times\mathfrak{M}^B(\Xi)\rightarrow
\mathscr{S}^\prime(\Xi).
\eeq

Arguments similar to those in the previous subsection allow us to obtain the following statement:

\smallskip
\begin{proposition}\label{P-izo-Moyal}
If the vector potential $A\in\Apol$, then the application $\Op^A$ defines a linear and topological isomorphism 
\beq\label{F-ext-Op-2}
\Op^A:\mathfrak{M}^B(\Xi)\overset{\sim}{\rightarrow}\mathcal{L}\big(\mathscr{S}(\X);\mathscr{S}(\X)\big)\bigcap\mathcal{L}\big(\mathscr{S}^\prime(\X);\mathscr{S}^\prime(\X)\big).
\eeq
\end{proposition}

Using Proposition \ref{P-izo-Sprime}, for any tempered distribution $F\in\mathscr{S}^\prime(\Xi)$ we can consider the restriction of $\Op^A(F)$ to $L^2(\X)$ (considering that any class in $L^2$ defines a unique tempered distribution).

\smallskip
\begin{definition}\label{D-bd-obs}
Given a vector potential $A\in\Apol$, we define \emph{the algebra of bounded magnetic symbols} associated to the magnetic field $B=dA$ as
\beq\label{FD-bd-obs}
\mathfrak{C}^B(\Xi):=\left\{F\in\mathscr{S}^\prime(\Xi)\,\mid\,\Op^A\big[L^2(\X)\big]\subset L^2(\X)\right\}.
\eeq
\end{definition}

By the Uniform Boundedness Principle (\cite{RS1}) $F\in\mathfrak{C}^B(\Xi)$ if and only if $\Op^A(F)\in\mathbb{B}\big(L^2(\X)\big)$ and using Proposition \ref{P-gauge-cov} we see that this condition only depends on the magnetic field $B=dA$.

\smallskip
\begin{proposition}\label{P-Moyal-H-symb} \textit{(Lemma 2.1 in \cite{IMP1})}\\
Given a magnetic field $B\in\Lpol$ we have that  $S^p(\Xi,\X)\subset\mathfrak{M}^B(\Xi)$ for any $p\in\mathbb{R}$.
\end{proposition} 

\smallskip
\begin{remark}\label{R-bd-obs}
We may transport on $\mathfrak{C}^B(\Xi)$ the operatorial norm from $\mathbb{B}\big(L^2(\X)\big)$ by $\|F\|_B:=\|\Op^A(F)\|_{\mathbb{B}(L^2(\X))}$. This norm only depends on $B$, due to Proposition \ref{P-gauge-cov}. Then $\mathfrak{C}^B(\Xi)$ becomes a $C^*$-algebra.
\end{remark}

\paragraph{The weak operator topology on $\mathfrak{C}^B(\Xi)$.}

We suppose fixed a magnetic field $B\in\Lpol$.

\begin{definition}
On $\mathfrak{C}^B(\Xi)$ we consider the locally convex topology $\mathfrak{T}^B_{\text{\sf opw}}$ defined by the semi-norms
$$
\rho^A_{u,v}(F)\,:=\,\left\langle u,\Op^A(F)v\right\rangle_{L^2(\X)},
$$
indexed by the pairs $(u,v)\in\big[L^2(\X)\big]^2$.
\end{definition}

\smallskip
\begin{proposition}
Taking into account that $\mathfrak{C}^B(\Xi)\subset\mathscr{S}^\prime(\Xi)$, the topology $\mathfrak{T}^B_{\text{\sf opw}}$ coincides on bounded subsets of $\mathfrak{C}^B(\Xi)$ (with respect to the $\|\cdot\|_B$-topology) with the \textit{weak distribution topology} induced from $\mathscr{S}^\prime(\Xi)$.
\end{proposition}

\section{Observables in bounded smooth magnetic fields.}

In the estimations that follow we shall often use operators of the form $\left<\nabla\right>^s$ (Fourier transforms of symbols of type $\mathfrak{p}_s$) and we shall frequently prefer to work with differential operators, so that we shall when possible consider orders of the form $s=2N$ with $N\in\mathbb{N}$ even if this will give slightly weaker estimations. We shall use the notation $\tilde{d}:=2[d/2]+2$ and for any $p\in\mathbb{R}_+$ we set $\tilde{p}:=2[(d+p)/2]+2$.

\subsection{Composition of H\"{o}rmander type symbols.}

In \cite{IMP1} we have proven a result concerning the composition of H\"{o}rmander type symbols (Proposition 2.6) that is very similar to the known result for Weyl calculus. We present here a new proof of a simplified version (that is enough for the applications to quantum mechanics that we have in view), emphasizing on the dependence of the estimations on the behaviour of the magnetic field. In fact we present this result in the following Theorem that is a direct consequence of the Proposition \ref{P-Ap-1} that we prove in the second Appendix to this paper.

\smallskip
\begin{theorem}\label{T-bound}
Given a magnetic field $B\in\Lbc$, for any pair $(p_1,p_2)\in\mathbb{R}^2$ the restriction of the Moyal product to $S^{p_1}(\Xi,\X)\times S^{p_2}(\Xi,\X)$ defines a continuous bilinear application
$$
S^{p_1}(\Xi,\X)\times S^{p_2}(\Xi,\X)\ni(F,G)\mapsto F\sharp^BG\in S^{p_1+p_2}(\Xi,\X).
$$
More precisely, for any pair $(q_1,q_2)\in\mathbb{N}\times\mathbb{N}$ there exists a constant $C:=C(d,p_1,p_2,q_1,q_2)$ such that
$$
\nu^{p_1+p_2-q_2}_{q_1,q_2}\big(F\sharp^BG\big)\ \leq\ C\,\mathring{\mathfrak{w}}_{q_1+\tilde{p}_1+\tilde{p}_2}(B)\big[\nu^{p_1}_{q_1+\tilde{p}_2,q_2+m_2}(F)\big]\big[\nu^{p_2}_{q_1+\tilde{p}_1,q_2+m_1}(G)\big],
$$
where $m_1=2[\tilde{p}_2+(q_1+\tilde{p}_1)/2]+2$ and $m_2=2[\tilde{p}_1+(q_1+\tilde{p}_2)/2]+2$.
\end{theorem}

\smallskip
We formulate separately a consequence of the above result that will be used several times in this paper.

\smallskip
\begin{proposition}\label{P-2-nd-symb-comp}
Given a magnetic field $B\in\Lbc$, for any pair $(F,G)\in S^{p_1}(\Xi,\X)\times S^{p_2}(\Xi,\X)$ with $(p_1,p_2)\in\mathbb{R}^2$ the following is true
$$
\mathfrak{R}(F,G):=F\sharp^BG-FG\in S^{p_1+p_2-1}(\Xi,\X).
$$ 
Moreover, for anny $(q_1,q_2)\in\mathbb{N}\times\mathbb{N}$ there exists a constant $C:=C(d,p_1,p_2,q_1,q_2)>0$ such that
\beq
\begin{split}
\nu^{p_1+p_2-1-q_2}_{q_1,q_2}\big(F\sharp^BG-FG\big)\,&\leq\,C\mathring{\mathfrak{w}}_{q_1+\tilde{p}_1
+\tilde{p}_2}(B)\mathring{\mathfrak{w}}_{2+q_1+\tilde{p}_1+\tilde{p}_2}(B)\\
&\times\Big(\rho^{p_1}_{q_1+\tilde{p}_2,q_2+m_2}(\nabla_xF)+\rho^{p_1-1}_{q_1+\tilde{p}_2,q_2+m_2}(\nabla_\xi F)\Big)\\
&\times\Big(\rho^{p_2-1}_{q_1+\tilde{p}_1,q_2+m_1}(\nabla_\xi G)+\rho^{p_2}_{q_1+\tilde{p}_1,q_2+m_1}(\nabla_x G)\Big),
\end{split}
\eeq
where $m_1=2[\tilde{p}_2+(q_1+\tilde{p}_1)/2]+2$ and $m_2=2[\tilde{p}_1+(q_1+\tilde{p}_2)/2]+2$.
\end{proposition}

\begin{proof}
We come back to formula \eqref{F-sharp-B-1} and notice that
\beq
\begin{split}
\big(F\sharp^B G\big)(X)&-F(X)G(X)=\\
&=\pi^{-2d}\!\int_{\Xi}\!dY\!\int_{\Xi}\!dZ\,e^{-2i(<\eta,z>-<\zeta,y>)}\varomega^B_x(y,z)\,F(X-Y)\,G(X-Z)-F(X)G(X)=
\end{split}
\eeq
$$
=\pi^{-2d}\!\int_{\Xi}dY\!\int_{\Xi}dZ\,e^{-2i(<\eta,z>-<\zeta,y>)}\varomega^B_x(y,z)\left[\int_0^1ds\big((Y\cdot\nabla)F\big)(X-sY)\right]\left[\int_0^1dt\big((Z\cdot\nabla)G\big)(X-tZ)\right]=
$$
$$
=-(2i\pi^{2d})^{-1}\!\int_{\Xi}\!dY\!\int_{\Xi}\!dZ\,e^{-2i(<\eta,z>-<\zeta,y>)}\varomega^B_x(y,z)\left[\int_0^1ds\big(\nabla_x F\big)(X-sY)\right]\cdot\nabla_\zeta\left[\int_0^1dt\big((Z\cdot\nabla)G\big)(X-tZ)\right]+
$$
$$
+(2i\pi^{2d})^{-1}\!\int_{\Xi}\!dY\!\int_{\Xi}\!dZ\,e^{-2i(<\eta,z>-<\zeta,y>)}\left[\int_0^1ds\big(\nabla_\xi F\big)(X-sY)\right]\cdot\nabla_z\left\{\varomega^B_x(y,z)\left[
\int_0^1dt\big((Z\cdot\nabla)G\big)(X-tZ)\right]\right\}=
$$
$$
=-(2i\pi^{2d})^{-1}\!\int_{\Xi}\!dY\!\int_{\Xi}\!dZ\,e^{-2i(<\eta,z>-<\zeta,y>)}\varomega^B_x(y,z)\left[\int_0^1ds\big(\nabla_x F\big)(X-sY)\right]\cdot\left[\int_0^1dt\big(\nabla_\xi G\big)(X-tZ)\right]+
$$
$$
+(4\pi^{2d})^{-1}\!\int_{\Xi}\!dY\!\int_{\Xi}\!dZ\,e^{-2i(<\eta,z>-<\zeta,y>)}\varomega^B_x(y,z)\left[\int_0^1sds\big(\nabla_\xi\nabla_x F\big)(X-sY)\right]\cdot\left[\int_0^1tdt\big(\nabla_x\nabla_\xi G\big)(X-tZ)\right]-
$$
$$
-(4\pi^{2d})^{-1}\int_{\Xi}\!dY\!\int_{\Xi}\!dZ\,e^{-2i(<\eta,z>-<\zeta,y>)}\varomega^B_x(y,z)\left[\int_0^1sds\big(\nabla_x\nabla_x F\big)(X-sY)\right]\cdot\left[\int_0^1tdt\big(\nabla_\xi\nabla_\xi G\big)(X-tZ)\right]+
$$
$$
+(4\pi^{2d})^{-1}\!\int_{\Xi}\!dY\!\int_{\Xi}\!dZ\,e^{-2i(<\eta,z>-<\zeta,y>)}\big[\big(\nabla_y\varomega^B_x\big)(y,z)\big]\left[\int_0^1ds\big(\nabla_x F\big)(X-sY)\right]\cdot\left[\int_0^1tdt\big(\nabla_\xi\nabla_\xi
G\big)(X-tZ)\right]+
$$
$$
+(2i\pi^{2d})^{-1}\!\int_{\Xi}\!dY\!\int_{\Xi}\!dZ\,e^{-2i(<\eta,z>-<\zeta,y>)}\varomega^B_x(y,z)\left[\int_0^1ds\big(\nabla_\xi F\big)(X-sY)\right]\cdot\left\{\left[\int_0^1dt\big(\nabla_x G\big)(X-tZ)\right]\right\}-
$$
$$
-(4\pi^{2d})^{-1}\!\int_{\Xi}\!dY\!\int_{\Xi}\!dZ\,e^{-2i(<\eta,z>-<\zeta,y>)}\varomega^B_x(y,z)\left[\int_0^1sds\big(\nabla_\xi\nabla_\xi F\big)(X-sY)\right]\cdot\left\{\left[\int_0^1tdt\big(\nabla_x\nabla_xG\big)(X-tZ)\right]\right\}+
$$
$$
+(4\pi^{2d})^{-1}\!\int_{\Xi}\!dY\!\int_{\Xi}\!dZ\,e^{-2i(<\eta,z>-<\zeta,y>)}\varomega^B_x(y,z)\left[\int_0^1sds\big(\nabla_x\nabla_\xi F\big)(X-sY)\right]\cdot\left\{\left[\int_0^1tdt\big(\nabla_\xi\nabla_xG\big)(X-tZ)\right]\right\}+
$$
$$
-(4\pi^{2d})^{-1}\!\int_{\Xi}\!dY\!\int_{\Xi}\!dZ\,e^{-2i(<\eta,z>-<\zeta,y>)}\big[\big(\nabla_y\varomega^B_x\big)(y,z)\big]\left[\int_0^1ds\big(\nabla_\xi F\big)(X-sY)\right]\cdot\left\{\left[
\int_0^1tdt\big(\nabla_\xi\nabla_xG\big)(X-tZ)\right]\right\}+
$$
$$
+(4\pi^{2d})^{-1}\!\int_{\Xi}\!dY\!\int_{\Xi}\!dZ\,e^{-2i(<\eta,z>-<\zeta,y>)}\left[\int_0^1sds\big(\nabla_\xi\nabla_\xi F\big)(X-sY)\right]\cdot\big[\big(\nabla_z\varomega^B_x\big)(y,z)\big]\left[
\int_0^1dt\big(\nabla_x G\big)(X-tZ)\right]=
$$
$$
-(4\pi^{2d})^{-1}\!\int_{\Xi}\!dY\!\int_{\Xi}\!dZ\,e^{-2i(<\eta,z>-<\zeta,y>)}\left[\int_0^1sds\big(\nabla_x\nabla_\xi F\big)(X-sY)\right]\cdot\big[\big(\nabla_z\varomega^B_x\big)(y,z)\big]\left[
\int_0^1dt\big(\nabla_\xi G\big)(X-tZ)\right]=
$$
$$
+(4\pi^{2d})^{-1}\!\int_{\Xi}\!dY\!\int_{\Xi}\!dZ\,e^{-2i(<\eta,z>-<\zeta,y>)}\left[\int_0^1ds\big(\nabla_\xi F\big)(X-sY)\right]\cdot\big[\big(\nabla_y\nabla_z\varomega^B_x\big)(y,z)\big]\left[
\int_0^1dt\big(\nabla_\xi G\big)(X-tZ)\right].
$$

Using once again Proposition \ref{P-Ap-1} with $\Theta\in BC^\infty\big(\X;C^\infty_{\text{\sf pol}}(\X\times\X)\big)$ equal to various derivatives of $\varomega^B$, we obtain 
\beq
\begin{split}
\nu^{p_1+p_2-1-q_2}_{q_1,q_2}\big(F\sharp^BG-FG\big)\,&\leq\,C\mathring{\mathfrak{w}}_{q_1+
\tilde{p}_1+\tilde{p}_2}(B)\mathring{\mathfrak{w}}_{2+q_1+\tilde{p}_1+
\tilde{p}_2}(B)\\
&\times\Big(\rho^{p_1}_{q_1+\tilde{p}_2,q_2+m_2}(\nabla_xF)
+\rho^{p_1-1}_{q_1+\tilde{p}_2,q_2+m_2}(\nabla_\xi F)\Big)\\
&\times\Big(
\rho^{p_2-1}_{q_1+\tilde{p}_1,q_2+m_1}(\nabla_\xi G)+
\rho^{p_2}_{q_1+\tilde{p}_1,q_2+m_1}(\nabla_x G)\Big).
\end{split}
\eeq
\end{proof}

\subsection{A criterion for $L^2$-boundedness.}

By the Schur-Holmgren criterion for boundedness of integral operators on $L^2$, we easily conclude that
\beq
\big\|\Op^A(F)\big\|_{\mathbb{B}(L^2(\X))}\,\leq\,\underset{x\in\X}{\sup}\big\|\mathfrak{K}^A_F(x,\cdot)\big\|_{L^1(\X)}.
\eeq
Using then \eqref{F-int-kernel} and Proposition 1.3.6 in \cite{ABG}, we obtain the following statement.

\smallskip
\begin{proposition}
Suppose that $B\in\Lbc$. Then $S^-(\Xi,\X)\subset\mathfrak{C}^B(\Xi)$, i.e. $\Op^A(F)\in\mathbb{B}\big(L^2(\X)\big)$ for any $F\in S^-(\Xi,\X)$ and
$$
\|F\|_B\equiv\big\|\Op^A(F)\big\|_{\mathbb{B}(L^2(\X))}\,\leq\,\nu^{-s}_{0,\tilde{d}}(F)
$$
for $s>0$ such that $F\in S^{-s}(\Xi,\X)$.
\end{proposition}

\smallskip
For the last estimation in the above Proposition one can see Lemma A.4 in \cite{MPR}.

\subsection{Inferior semiboundedness.}\label{Sss-inf-sbd}

\begin{proposition}\label{P-inf-sbd}
Suppose that $B\in\Lbc$ and $F\in S^p(\Xi,\X)$, with $p\geq0$ and $F$ elliptic if $p>0$, verifies $F\geq a_F>0$ for some $a_F\in\mathbb{R}_+$. 
Then there exist $G\in S^{p/2}(\Xi,\X)$ and $X\in S^0(\Xi,\X)$ such that $F=G\sharp^BG\,+\,X$ and the operator norm of $\Op^B(X)$ is bounded by a constant defined by the product of a polynomial $\mathscr{W}_p(B)$ of maximum degree $2[p]+2$ in the weights $\mathring{\mathfrak{w}}_M(B)$ with $M\in\mathbb{N}$ depending only on $d$ and $p$, with positive coefficients depending only on $d$ and $p$, and a polynomial $\mathscr{N}_p(F)$ of maximum degree $2[p]+2$  in the seminorms $\nu^{P_1}_{M_1,M_2}(F)$ with $P_1$, $M_1$ and $M_2$ depending only on $d$ and $p$, with positive coefficients depending only on $d$, $p$ and $a_F^j>0$, with $0\leq j\leq 2[p]+2$.
\end{proposition}

\begin{proof}
We can define $G_0:=\sqrt{F}\in S^{p/2}(\Xi,\X)$ and notice that  (with the notation in Proposition \ref{P-2-nd-symb-comp}):
$$
\big(F-G_0\sharp^BG_0\big)(X)=G_0(X)^2-\big(G_0\sharp^BG_0\big)(X)=\mathfrak{R}(G_0,G_0)(X).
$$
We apply Proposition \ref{P-2-nd-symb-comp} in order to conclude that:
$$
X^B_1(F):=F-G_0\sharp^BG_0\ \in\ S^{p-1}(\Xi,\X),
$$
$$
\nu^{p-1-q_2}_{q_1,q_2}\big(X^B_1(F)\big)\,\leq\,C\mathring{\mathfrak{w}}_
{q_1+2\tilde{p}}(B)\mathring{\mathfrak{w}}_{2+q_1+2\tilde{p}}(B)\nu^{p/2}_{q_1+\tilde{p},q_2+m}(\nabla_x\sqrt{F})\nu^{p/2-1}_{q_1+\tilde{p},q_2+m}(\nabla_\xi\sqrt{F})
$$
with $m:=2[(q_1+3n)/2]+2$. We clearly have $G_0=\sqrt{F}\geq\sqrt{a_F}>0$ and we can define
$$
G_1:=(1/2)G_0^{-1}X^B_1(F)\,\in\,S^{p/2-1}(\Xi,\X)
$$
and notice that
$$
F-\big(G_0+G_1\big)\sharp^B\big(G_0+G_1\big)\,=\,F-G_0\sharp^BG_0-G_1\sharp^BG_1-\big(G_0\sharp^BG_1+G_1\sharp^B
G_0\big)
$$
$$
=-G_1\sharp^BG_1-\big(G_0\sharp^BG_1+G_1\sharp^BG_0-2G_0G_1\big)\,=:\,X^B_2(F)\,\in\,S^{p-2}(\Xi,\X)
$$
with
\beq\nu^{p-2-q_2}_{q_1,q_2}\big(X^B_2(F)\big)\,\leq\,
C\mathring{\mathfrak{w}}_{q_1+2\tilde{p}}(B)\nu^{p/2}_{q_1+\tilde{p},q_2+m}(\nabla_x G_1)\nu^{p/2-1}_{q_1+\tilde{p},q_2+m}(\nabla_\xi G_1)
\eeq
$$
+C\mathring{\mathfrak{w}}_{q_1+2\tilde{p}}(B)\mathring{\mathfrak{w}}_{2+q_1+2\tilde{p}}(B)\big[\nu^{p/2}_{q_1+\tilde{p},q_2+m}(\nabla_x G_0)
\nu^{p/2-1}_{q_1+\tilde{p},q_2+m}(\nabla_\xi G_1)+\nu^{p/2-1}_{q_1+\tilde{p},q_2+m}(\nabla_\xi G_0)\nu^{p/2}_{q_1+\tilde{p},q_2+m}(\nabla_x G_1)\big].
$$
Let us set $n_p:=[p]+1\in\mathbb{N}$ and define recursively for $1\leq k\leq n_p$
\beq
\left\{
\begin{array}{l}
X^B_k\,:=\,F-\left(\underset{0\leq j\leq k-1}{\sum}G_j\right)\sharp^B\left(\underset{0\leq j\leq k-1}{\sum}G_j\right)\,\in\,S^{p-n_p}(\Xi,\X)\,\subset\,S^{0}(\Xi,\X),\\
G_k\,:=\,(1/2)G_0^{-1}X^B_k\,\in\ S^{(p/2)-n_p}(\Xi,\X).
\end{array}
\right.
\eeq
Using the above results we obtain
$$
\nu^{p-n_p-q_2}_{q_1,q_2}\big(X^B_{n_p}\big)\,\leq\,\mathscr{W}_p(B)\,\mathscr{N}_p(F),
$$
where $\mathscr{W}_p(B)$ is a polynomial of maximum degree $2n_p$, with positive coefficients depending only on $d$, $p$, $q_1$ and $q_2$ in the weights $\mathring{\mathfrak{w}}_M(B)$ with $q_1+2\tilde{p}\leq M\leq 2+q_1+\tilde{p}n_p$ and $\mathscr{N}_p(F)$ is a polynomial of maximum degree $2n_p$, with positive coefficients depending only on $d$, $p$, $q_1$, $q_2$ and $a_F^j>0$ (with $0\leq j\leq n_p$) in the semi-norms $\nu^{P_1}_{M_1,M_2}(F)$ with $0\leq P_1\leq p$, $q_1+\tilde{p}\leq M_1\leq q_1+\tilde{p}n_p$ and $q_2+m\leq M_2\leq q_2+mn_p$.
\end{proof}

Finally notice that the above Proposition implies for any $\phi\in\mathscr{S}(\X)$ the estimation
$$
\left\langle\phi,\Op^A(F)\phi\right\rangle_{L^2(\X)}\,\geq-\left\|\Op^A(X^B_{n_p})\right\|_{\mathbb{B}(L^2(\X))}\,\|\phi\|^2_{L^2(\X)},
$$
and one concludes that
$$
\Op^A(F)\,+\,\left\|\Op^A(X^B_{n_p})\right\|_{\mathbb{B}(L^2(\X))}\bb1\,+\,a_F\bb1\ \geq\ a_F\bb1\ >\ 0
$$
and also
$$
\Op^A(F)\,+\,\big(\mathscr{W}_p(B)\cdot\mathscr{N}_p(F)\,+\,a_F\big)\bb1\ \geq\ a_F\bb1\ >\ 0.
$$

\subsection{A Calderon-Vaillancourt type Theorem.}

In this subsection we give another proof for Theorem 3.1 in \cite{IMP1} for symbols in our restricted class $S^0(\Xi,\X)$. This proof is inspired by the one in \cite{H-I} for the case of the usual Weyl calculus.

\smallskip
\begin{theorem}
For any magnetic field $B\in\Lbc$ we have $S^0(\Xi,\X)\subset\mathfrak{C}^B(\Xi)$. More precisely, for any $F\in S^0(\Xi,\X)$ we have
$$
\|F\|_B\equiv\big\|\Op^A(F)\big\|_{\mathbb{B}(L^2(\X))}\leq C(d).
$$
\end{theorem}

\begin{proof}
From the definition of the class $S^0(\Xi,\X)$ in Definition \ref{D-H-symb} we conclude that given any $F\in S^0(\Xi,\X)$ there exists $M_F>0$ such that $|F(X)|\leq M_F$ for any $X\in\Xi$. Then $\widetilde{F}:=(M_F+\delta)^2-F^2$ is a strictly positive symbol of class $S^0(\Xi,\X)$ for any $\delta>0$. 
For any $\phi\in\mathscr{S}(\X)$ we can compute
$$
\big\|\Op^A(F)\phi\big\|_{L^2(\X)}^2=\left\langle\Op^A(F)\phi\,,\,\Op^A(F)\phi\right\rangle_{\!L^2(\X)}=\left\langle\phi\,,\,\Op^A(F\sharp^BF)\phi\right\rangle_{\!L^2(\X)}=
$$
$$
=\left\langle\phi\,,\,\Op^A(F\sharp^BF-F^2)\phi\right\rangle_{L^2(\X)}+(M_F+\delta)^2\|\phi\|_{\!L^2(\X)}^2-\left\langle\phi\,,\,\Op^A(\widetilde{F})\phi\right\rangle_{\!L^2(\X)}
$$
We can apply Proposition \ref{P-inf-sbd} above (with $n_p\equiv n_0=1$) and deduce that there exists some symbol $G^B_F\in S^0(\Xi,\X)$ and some symbol $X^B_F\in S^{-1}(\Xi,\X)$ such that
$$
\widetilde{F}=G^B_F\sharp^BG^B_F\,+\,X^B_F,\qquad\big\|X^B_F\big\|_B\,\leq\,C(d)\mathring{\mathfrak{w}}_{2\tilde{d}}(B)\mathring{\mathfrak{w}}_{2\tilde{d}+2}(B)\big(\mu^{-1}_{\tilde{d},\tilde{d}+m}(F)\big)^2
$$
with $m:=2[3\tilde{d}/2]+2$. Then we use Proposition \ref{P-2-nd-symb-comp} in order to get $F^2-F\sharp^BF\in S^{-1}(\Xi,\X)$ and the estimations
$$
\nu^{-1-q_2}_{q_1,q_2}\big(F^2-F\sharp^BF\big)\leq C\mathring{\mathfrak{w}}_{q_1+2\tilde{d}}(B)\mathring{\mathfrak{w}}_{2+q_1+2\tilde{d}}(B)
\nu^0_{q_1+\tilde{d},q_2+m}(\nabla_xF)\nu^{-1}_{q_1+\tilde{d},q_2+m}(\nabla_\xi F),
$$ 
with $m:=2[(q_1+3\tilde{d})/2]+2$ and
$$
\|F\sharp^BF-F^2\|_B\leq\nu^{-1}_{0,\tilde{d}}(F\sharp^BF-F^2)\leq C\mathring{\mathfrak{w}}_{2\tilde{d}}(B)\mathring{\mathfrak{w}}_{2+2\tilde{d}}(B)
\nu^0_{\tilde{d},\tilde{d}+m}(\nabla_xF)\nu^{-1}_{\tilde{d},\tilde{d}+m}(\nabla_\xi F)
$$
with $m:=2[3\tilde{d}/2]+2$. Thus for any $\phi\in\mathscr{S}(\X)$ one has
\begin{align}
\big\|\Op^A(F)\phi\big\|_{L^2(\X)}^2&=(M_F+\delta)^2\|\phi\|_{L^2(\X)}^2-\big\|\Op^A(G^B_F)\phi\big\|_{L^2(\X)}^2-\left\langle\phi\,,\,\Op^A\big(X^B_F+(F\sharp^BF-F^2)\big)\phi\right\rangle_{L^2(\X)}\nonumber\\
&\leq\Big((M_F+\delta)^2+C(d)\mathring{\mathfrak{w}}
_{2\tilde{d}}(B)\mathring{\mathfrak{w}}_{2\tilde{d}+2}(B)\big(\mu^{-1}_{\tilde{d}+1,\tilde{d}+m+1}(F)\big)^2\Big)\|\phi\|_{L^2(\X)}^2,\ \forall\delta>0\nonumber.
\end{align}
\end{proof}

\subsection{Self-adjointness.}

It is well known that the physical observables of a quantum system with configuration space $\X$, in a magnetic field $B\in\bigwedge^2\X$, are described by self-adjoint operators acting in the Hilbert space $L^2(\X)$. We remark that \textit{any real symbol in $\mathfrak{C}^B(\Xi)$ defines a bounded physical observable}. 

\smallskip
In order to study \textit{unbounded physical observables} we have to pay attention to the domain of definition of magnetic quantized operators. A reasonable choice for \textit{a real symbol} $F\in\mathfrak{M}^B(\Xi)$ could be (\textit{the maximal operator})
\beq\label{F-def-dom}
\mathcal{D}^A_F\ :=\ \left\{f\in L^2(\X)\,\mid\,\Op^A(F)f\in L^2(\X)\right\},
\eeq
where we take into account that $L^2(\X)$ may be identified with a complex linear subspace of $\mathscr{S}^\prime(\X)$. We notice that
\beq\label{F-def-dom-2}
\mathcal{D}^A_F= \left\{f\in L^2(\X)\,\mid\,\exists c_f>0,\,\left|\big\langle f,\Op^A(F)\phi\big\rangle_{L^2(\X)}\right|\,\leq c_f\|\phi\|_{L^2(\X)}\,\forall\phi\in\mathscr{S}(\X)\right\}.
\eeq

Due to the fact that $F\in\mathfrak{M}^B(\Xi)$ we know that $\Op^A(F)$ leaves $\mathscr{S}(\X)$ invariant (by the definition in \eqref{FD-Moyal-alg}), so that $\mathscr{S}(\X)\subset\mathcal{D}^A_F$, and we conclude that the domain $\mathcal{D}^A_F$ is a dense linear subspace of $L^2(\X)$. The problem is that it is not clear if with this domain of definition the operator $\Op^A(F)$ is symmetric! In fact its adjoint, having by definition the domain
\beq\label{F-dom-adj}
\widetilde{\mathcal{D}}^A_F\ :=\ \left\{f\in L^2(\X)\,\mid\,\exists c_f>0,\,\left|\big\langle f,\mathfrak{Op}^A(F)h\big\rangle_{L^2(\X)}\right|\,\leq c_f\|h\|_{L^2(\X)}\,\forall h\in\mathcal{D}^A_F\right\},
\eeq
will in principle be only \textit{a restriction} of it!

\smallskip
We could also consider just the operator 
$$
\Op^A(F):\mathscr{S}(\X)\rightarrow L^2(\X)
$$
and notice that it is by definition symmetric and thus closable and just take its closure (\textit{the minimal operator}).
Thus, in order to have the self-adjointness of our operator, it is enough to prove that the domain $\mathcal{D}^A_F$ is the closure of $\mathscr{S}(\X)$ for the \textit{graph-norm} of the operator $\mathfrak{Op}^A(F)$, i.e. to prove that
\beq\label{F-self-adj}
\mathcal{D}^A_F\ni f\ \Longleftrightarrow\ \exists\{\phi_n\}_{n\in\mathbb{N}}\subset\mathscr{S}(\X),\ \underset{n\rightarrow\infty}{\lim}\big(\|f-\phi_n\|_{L^2(\X)}^2+\|\mathfrak{Op}^A(F)(f-\phi_n)\|_{L^2(\X)}^2\big)=0.
\eeq

A procedure to prove self-adjointness in $L^2(\X)$ for an operator of the form $\Op^A(F)$ for some real symbol $F\in\mathfrak{M}^B(\Xi)$ is to construct \textit{a resolvent} for it. More precisely, to prove existence of two symbols
\beq\label{FD-rez}
\r^B_\pm(F)\in\mathfrak{C}^B(\Xi)\bigcap\mathfrak{M}
^B(\Xi)
\eeq 
such that
\beq\label{Frez-cond}
(F\mp i)\sharp^B\r^B_\pm(F)\,=\,1\qquad \r^B_\pm(F)\sharp^B(F\mp i)\,=\,1.
\eeq

In \cite{IMP1} we define a family of dense domains in $L^2(\X)$ that are domains of self-adjointness for a large class of magnetic quantizations of real symbols. Let us briefly recall these facts.

\smallskip
\begin{definition}\label{D-magn-Sob-sp}
Given a vector potential $A\in\Apol$, for any $s\in\mathbb{R}_+$ we define \textit{the magnetic Sobolev space} of order $s$ as
\beq\label{FD-Sob-sp}
\mathscr{H}^s_A(\X)\,:=\,\mathcal{D}^A_{\mathfrak{p_s}},\quad\mathfrak{p}_s(x,\xi):=<\xi>^s
\eeq
endowed with the graph norm that generates a scalar product
\beq\label{FD-magn-Sob-scprod}
\langle f,g\rangle_{\mathscr{H}^s_A}\ :=\ \langle f,g\rangle_{L^2(\X)}+\langle\mathfrak{Op}^A(\mathfrak{p}_s)f,\mathfrak{Op}^A(\mathfrak{p}_s)g\rangle_{L^2(\X)},
\eeq 
for which $\mathscr{H}^s_A(\X)$ is complete and thus a Hilbert space.
\end{definition}

\smallskip
For the rest of this section we shall suppose that $B\in\Lbc$ and make use of the notation and results in the first Appendix A.1. The following Theorem contains the main results in Theorem 5.1 in \cite{IMP1} and Proposition 6.31 in \cite{IMP2} and we present here a new proof of these results developping the ideas in the proof of Theorem 1.8 in \cite{MPR}.

\smallskip
\begin{theorem}\label{T-self-adj} 
Given  a magnetic field $B\in\Lbc$ and a choice of a vector potential $A\in\Apol$, for any real elliptic symbol $F\in S^p(\Xi,\X)$ with $p>0$ we have that:
\begin{enumerate}
\item there exist some symbols $\r^B_\pm(F)\in\mathscr{S}^\prime(\Xi)$ such that $(F\mp i)\sharp^B\r^B_\pm(F)=1$ and
$\r^B_\pm(F)\sharp^B(F\mp i)=1$;
\item $\r^B_\pm(F)\in S^{-p}(\Xi,\X)$;
\item $\Op^A(F)$ is self-adjoint in $L^2(\X)$ with domain $\mathscr{H}^p_A(\X)$ and essentialy self-adjoint on $\mathscr{S}(\X)$.
\end{enumerate}
\end{theorem}

\begin{proof}
The ellipticity condition on $F\in S^p(\Xi,\X)$ with $p>0$ means that there exist two constants $R>0$ and $C>0$ such that for $|\xi|\geq R$ we have the bound
$C<\xi>^p\,\leq F(x,\xi)$.
Then let us fix some $a>0$ large enough such that $F+a>0$, set $F_a:=F+a$ and compute
$$
\big(F_a\sharp^BF_a^{-1}\big)(X)-1=\pi^{-2d}\int_{\Xi}dY\int_{\Xi}dZ\,e^{-2i(<\eta,z>-<\zeta,y>)}\varomega^B_x(y,z)\frac{F(X-Y)+a}{F(X-Z)+a}\,-\,1=
$$
\beq\label{F-cor-rez-2}
=\,\pi^{-2d}\int_{\Xi}dY\int_{\Xi}dZ\,e^{-2i(<\eta,z>-<\zeta,y>)}\varomega^B_x(y,z)\,\times
\eeq
\beq\label{F-cor-rez-3}
\times\,\underset{1\leq j\leq d}{\sum}\int_0^1d\tau\frac{(z_j-y_j)\big(\partial_{x_j}F\big)(X-Z+\tau(Z-Y))+(\zeta_j-\eta_j)\big(\partial_{\xi_j}F\big)(X-Z+\tau(Z-Y))}{F(X-Z)+a}.
\eeq

First let us get rid of the linear terms $z-y$ and $\zeta-\eta$ by integration by parts, using the identities:
$$
z_je^{-2i<\eta,z>}=(i/2)\partial_{\eta_j}e^{-2i<\eta,z>},\quad y_je^{2i<\zeta,y>}=(1/2i)\partial_{\zeta_j}e^{2i<\zeta,y>},
$$
$$
\eta_je^{-2i<\eta,z>}=(i/2)\partial_{z_j}e^{-2i<\eta,z>},\quad\zeta_je^{2i<\zeta,y>}=(1/2i)\partial_{y_j}e^{2i<\zeta,y>}.
$$
Integrating by parts the term containing the factor
$$
(z_j-y_j)e^{-2i(<\eta,z>-<\zeta,y>)}=-(1/2i)\big(\partial_{\eta_j}+\partial_{\zeta_j}\big)e^{-2i(<\eta,z>-<\zeta,y>)}
$$
we obtain integrals of the form
$$
\mathcal{I}_a(F):=(1/2i)\pi^{-2d}\underset{1\leq j\leq d}{\sum}\int_0^1d\tau
\int_{\Xi}dY\int_{\Xi}dZ\,e^{-2i\sigma(Y,Z)}\varomega^B_x(y,z)\big(\partial_{\eta_j}+
\partial_{\zeta_j}\big)\frac{\big(\partial_{x_j}F\big)(X-Z+\tau(Z-Y))}{F(X-Z)+a}
$$
and notice that
$$
\big(\partial_{\eta_j}+\partial_{\zeta_j}\big)\frac{\big(\partial_{x_j}F\big)(X-Z+\tau(Z-Y))}{F(X-Z)+a}=
$$
$$
=-\frac{\big(\partial_{\xi_j}\partial_{x_j}F\big)(X-Z+\tau(Z-Y))}{F(X-Z)+a}-\big(\partial_{x_j}F\big)(X-Z+\tau(Z-Y))\big(\partial_{\xi_j}F_a^{-1}\big)(X-Z).
$$
We have to take into account that $\partial_{x_j}F\in S^p(\Xi,\X)$, $\partial_{\xi_j}F\in S^{p-1}(\Xi,\X)$, $\partial_{\xi_j}\partial_{x_j}F\in S^{p-1}(\Xi,\X)$, $(F+a)^{-1}\in S^{-p}(\Xi,\X)$ and $(F+a)^{-2}\in S^{-2p}(\Xi,\X)$ and use the result in Proposition \ref{P-Ap-1} to obtain that $\mathcal{I}_a(F)\in S^0(\Xi,\X)$ with the following estimations on its seminorms:
$$
\nu^{-m}_{n,m}\left(\mathcal{I}_a(F)\right)\leq
\ C(d,p,n,m)\mathring{\mathfrak{w}}_{n+n_1+n_2}(B)\times
$$
\beq\label{F-rest-invers-1}
\times\underset{0\leq k\leq m}{\sum}\Big(\rho^{p-1-k}_{n+n_2+1;k+1,m_2}(F)\,\rho^{-(p-1-k+m)}_{n+n_1;m-k,m_1}(F_a^{-1})+\rho^{p-k}_{n+n_2;k,m_2}(F)\,\rho^{-(p-k+m)}_{n+n_1;m-k+1,m_1}(F_a^{-1})\Big),
\eeq
where $n_1=2[(d+p)/2]+2$, $n_2=2[(d-p)/2]+2$, $m_1=2[(n+n_1)/2+n_2]+2$ and $m_2=2[(n+n_2)/2+n_1]+2$. 

\smallskip
Let us study now the term containing the factor
$$
(\zeta_j-\eta_j)e^{-2i(<\eta,z>-<\zeta,y>)}=(1/2i)\big(\partial_{y_j}+\partial_{z_j}\big)
e^{-2i(<\eta,z>-<\zeta,y>)}.
$$
After integration by parts we obtain integrals of the form
$$
\mathcal{J}_a(F):=-(1/2i)\pi^{-2d}\underset{1\leq j\leq d}{\sum}\int_0^1\hspace*{-5pt}d\tau
\int_{\Xi}dY\int_{\Xi}dZ\,e^{-2i\sigma(Y,Z)}\big(\partial_{y_j}+
\partial_{z_j}\big)\Big(\varomega^B_x
(y,z)\frac{\big(\partial_{\xi_j}F\big)(X-Z+\tau(Z-Y))}{F(X-Z)+a}\Big)
$$
and notice that:
$$
\Big(\big(\partial_{y_j}+\partial_{z_j}\big)\varomega^B_x\Big)=\Big(\big(\partial_{y_j}+\partial_{z_j}\big)F^B_x\Big)(y,z)\varomega^B_x(y,z)
$$
$$
\big(\partial_{y_j}+\partial_{z_j}\big)\frac{\big(\partial_{\xi_j}F\big)(X-Z+\tau(Z-Y))}{F(X-Z)+a}=
$$
$$
=-\frac{\big(\partial_{\xi_j}\partial_{x_j}F\big)(X-Z+\tau(Z-Y))}{F(X-Z)+a}-\big(\partial_{\xi_j}F\big)(X-Z+\tau(Z-Y))\big(\partial_{x_j}F_a^{-1}\big)(X-Z).
$$
As in the previous analysis we obtain that $\mathcal{J}_a(F)\in S^0(\Xi,\X)$ with the following estimations on the seminorms giving the Fr\'{e}chet topology:
$$
\nu^{-m}_{n,m}\left(\mathcal{I}_a(F)\right)\leq
\ C(d,p,n,m)\mathring{\mathfrak{w}}_{n+n_1+n_2}(B)\times
$$
\beq\label{F-rest-invers-2}
\times\underset{0\leq k\leq m}{\sum}\Big(\rho^{p-1-k}_{n+n_2+1;k+1,m_2}(F)\,\rho^{-(p-1-k+m)}_{n+n_1;m-k,m_1}(F_a^{-1})+\rho^{p-k-1}_{n+n_2;k,m_2}(\nabla_\xi F)\,\rho^{-(p-1-k+m)}_{n+n_1+1;m-k,m_1}(F_a^{-1})\,+
\eeq
$$
+\,\rho_1(B)\rho^{p-k-1}_{n+n_2+1;k+1,m_2}(F)\,\rho^{-(p-1-k+m)}_{n+n_1+1;m-k,m_1}(F_a^{-1})\Big),
$$
where $n_1=2[(d+p)/2]+2$, $n_2=2[(d-p)/2]+2$, $m_1=2[(n+n_1)/2+n_2]+2$ and $m_2=2[(n+n_2)/2+n_1]+2$. 

\smallskip
First we notice that 
$$
\rho^{-(p-1+r)}_{n;r,m}(F_a^{-1})=\underset{|\gamma|=r}{\max}\rho^{-(p-1+r)}_{n;0,m}\big(\partial_\xi^\gamma F_a^{-1}\big)
$$
so that
$$
\rho^{-(p-1+r)}_{n;0,m}\big(\partial_\xi^\gamma F_a^{-1}\big)=
\underset{|\gamma^1|\leq n}{\max}\ 
\underset{|\gamma^2|\leq m}{\max}\ \underset{(x,\xi)\in\Xi}{\sup}\left<\xi\right>^{p-1+r}\big|\big(\partial^{\gamma^1}_x\partial^{\gamma^2+\gamma}_\xi F_a^{-1}\big)(x,\xi)\big|
$$
and using Fa\'{a} di Bruno's formula \cite{G} we can write
$$
\left<\xi\right>^{p-1+r}\big|\big(\partial^{\gamma^1}_x\partial^{\gamma^2+\gamma}_\xi F_a^{-1}\big)(x,\xi)\big|\leq\ \left<\xi\right>^{p-1+r}(\gamma^1!)((\gamma^2+\gamma)!)\times
$$
$$
\times \underset{1\leq l\leq n}{\sum}\frac{(-1)^l}{l!}\underset{1\leq k\leq m}{\sum}\frac{(-1)^k}{k!}\ \sum\limits_{\alpha^1+\cdots+\alpha^l=
\gamma^1}\sum\limits_{\beta^1+\cdots+\beta^k=
\gamma^2+\gamma}\left(\prod\limits_{s=1}^{l}\frac{1}{\alpha^s!}
\prod\limits_{t=1}^{k}\frac{1}{\beta^t!}
\left|\big(\partial^{\alpha^s}_x\partial^{\beta^t}
_\xi F\big)(X)\right|\right)\big|F(X)+a\big|^{-(l+k+1)}\leq
$$
$$
\leq (\gamma^1!)((\gamma^2+\gamma)!)\frac{\left<\xi\right>^{p-1}}{F_a}\underset{1\leq l\leq n}{\sum}\underset{1\leq k\leq m}{\sum}\frac{(-1)^{l+k}}{l!k!}\sum\limits_{\alpha^1+\cdots+\alpha^l=
\gamma^1}
\hspace*{-15pt}\left<\xi\right>^{|\gamma|}\hspace*{-15pt}
\sum\limits_{\beta^1+\cdots+\beta^k=\gamma^2+\gamma}\left(\prod\limits_{s=1}^{l}\frac{1}{\alpha^s!}
\prod\limits_{t=1}^{k}\frac{1}{\beta^t!}\frac{\left|\big(\partial^{\alpha^s}_x\partial^{\beta^t}_\xi F\big)(X)\right|}{F_a}\right)
$$
$$
\leq C(F,n,m,r)\underset{\xi\in\X^*}{\sup}\frac{\left<\xi\right>^{p-1}}{\left<\xi\right>^p+a}.
$$
Using the monotonicity and the concavity of the logarithm one can prove the following inequality:
$$
a+b\geq\big(q^{-1}a\big)^q\big((1-q)^{-1}b\big)^{(1-q)},\quad\forall(a,b,q)\in\mathbb{R}_+\times\mathbb{R}_+\times(0,1).
$$
Taking $b=\left<\xi\right>^p$ we obtain that 
$$
\frac{\left<\xi\right>^{p-1}}{\left<\xi\right>^p+a}\,=\,\frac{b^{1-1/p}}{a+b}\,\leq\,\frac{b^{1-1/p}}{\big(q^{-1}a\big)^q\big((1-q)^{-1}b\big)^{(1-q)}}\,\leq\,q\left(\frac{1-q}{q}\right)^{1-q}\frac{\left<\xi\right>^{pq-1}}{a^q}
$$
and thus, chosing some $q\in(0,p^{-1})\cap(0,1)$ we obtain that for any $r\in\mathbb{N}$ we have the estimation 
\beq\label{Est-snorms-inv}
\rho^{-(p-1+r)}_{n;r,m}(F_a^{-1})\ \leq\ C(F,n,m,p,q)\,a^{-q}.
\eeq
Using also \eqref{F-rest-invers-1} and \eqref{F-rest-invers-2} and denoting by 
\beq\label{FD-symbol-rest}
\mathfrak{x}^B_F(a):=F_a\sharp^BF_a^{-1}-1\in S^0(\Xi,\X),
\eeq
we obtain for any $q\in(0,p^{-1})\cap(0,1)$ that
\beq\label{Est-symb-rest}
\nu^{-m}_{n,m}\big(\mathfrak{x}^B_F(a)\big)\ \leq\ C(F,n,m,p,q)\,a^{-q}\,\mu_1(B)\mathring{\mathfrak{w}}_{n+n_1+n_2}(B).
\eeq

From Theorem 3.1 in \cite{IMP1} we know that there exist two numbers $(p_1,p_2)\in\mathbb{N}\times\mathbb{N}$ depending only on the dimension of $\X$ such that
$$
\|F\|_B\,\leq\,C(d)\,\nu^0_{p_1,p_2}(F),\qquad\forall F\in S^0(\Xi,\X).
$$
Thus, for $q=\min\{1,p^{-1}\}$ and $n_1=2[(d+p)/2]+2$, $n_2=2[(d-p)/2]+2$ we have
$$
\big\|\mathfrak{x}^B_F(a)\big\|_B\,\leq\,C(d)\,\nu^0_{p_1,p_2}\big(F_a\sharp^BF_a^{-1}-1\big)\,\leq\,C(d,F)\,a^{-q}\,\mu_1(B)\mathring{\mathfrak{w}}_{p_1+n_1+n_2}(B).
$$
In conclusion, if we choose 
$$
a\,>\,\big[2C(d,F)\mu_1(B)\mathring{\mathfrak{w}}_{p_1+n_1+n_2}(B)\big]^{1/q}
$$
we have $\big\|\mathfrak{x}^B_F(a)\big\|_B\leq1/2$. Moreover we notice that
$$
F_a\sharp^B\big[F_a^{-1}\sharp^B\big(1-\mathfrak{x}^B_F(a)\big)\big]=1+\mathfrak{x}^B_F(a)-\mathfrak{x}^B_F(a)-\big(
\mathfrak{x}^B_F(a)\sharp^B\mathfrak{x}^B_F(a)\big)=1-\big(\mathfrak{x}^B_F(a)\sharp^B\mathfrak{x}^B_F(a)\big).
$$
From these we may conclude that the following limit exists in $\mathfrak{C}^B(\Xi)$ in the topology of the norm $\|\cdot\|_B$
\beq\label{FD-ser-rest}
\mathfrak{z}^B_F(a):=1+\underset{N\nearrow\infty}{\lim}\sum\limits_{n=1}^{N}\big[-\mathfrak{x}^B_F(a)\big]^{\sharp^Bn}
\eeq 
and the symbol $\mathfrak{r}^B_F(a):=F_a^{-1}\sharp^B\mathfrak{z}^B_F(a)$ satisfies the equality
\beq\label{FD-symb-invers}
\big(F+a\big)\sharp^B\mathfrak{r}^B_F(a)=1.
\eeq
Starting then with the reversed product $F_a^{-1}\sharp^BF_a$ and repeating exactly the above arguments we obtain a left inverse for $F+a$ for the magnetic Moyal product, and due to the well known abstract argument they have to be equal. We conclude that $\Op^A(F)$ is a symmetric operator having the real number $-a$ in its resolvent set; as this set is open we can find in its resolvent set points with strictly positive and strictly negative imaginary parts so that we conclude that it is self-adjoint. Moreover we know that $\mathfrak{z}^B_F(a)\in\mathfrak{C}^B(\Xi)$ can be analytically continued to an analytic map
$$
\{\mathcal{z}\in\mathbb{C}\,\mid\,\Im\mathcal{m}\,\mathcal{z}\ne0\}\bigcup\{\mathcal{x}\in\mathbb{R}\,\mid\,\mathcal{x}+a<\epsilon\}\ni\mathcal{z}\mapsto\mathfrak{r}^B_F(
\mathcal{z})\in\mathfrak{C}^B(\Xi)
$$
for some $\epsilon>0$ small enough and this map verifies the \textit{resolvent equation}:
\beq\label{E-rez-symb}
\mathfrak{r}^B_F(\mathcal{z}_1)-\mathfrak{r}^B_F(
\mathcal{z}_2)=\big(\mathcal{z}_2-\mathcal{z}_1\big){r}^B_F(\mathcal{z}_1)\sharp^B\mathfrak{r}^B_F
(\mathcal{z}_2)=\big(\mathcal{z}_2-\mathcal{z}_1\big){r}^B_F(\mathcal{z}_2)\sharp^B\mathfrak{r}^B_F(\mathcal{z}_1)
\eeq
and also the defining relations for the inverse:
\beq\label{F-invers}
\big(F+\mathcal{z}\big)\sharp^B\mathfrak{r}^B_F(\mathcal{z})=
\mathfrak{r}^B_F(\mathcal{z})\sharp^B\big(F+\mathcal{z}\big)=1.
\eeq

Moreover we notice that
$$
\mathfrak{p}_p\sharp^B\mathfrak{r}^B_F(\mathcal{z})=\mathfrak{p}_p\sharp^B\mathfrak{r}^B_F(a)+
\big(a-\mathcal{z}\big)\mathfrak{p}_p\sharp^B\mathfrak{r}^B_F(a)\sharp^B\mathfrak{r}^B_F(\mathcal{z})=
$$
$$
=\mathfrak{p}_p\sharp^B(F+a)^{-1}\sharp\mathfrak{z}^B_F(a)+
\big(a-\mathcal{z}\big)\mathfrak{p}_p\sharp^B(F+a)^{-1}\sharp^B\mathfrak{z}^B_F(a)
\sharp^B\mathfrak{r}^B_F(\mathcal{z})\in\mathfrak{C}^B(\Xi),
$$
because $\mathfrak{p}_p\sharp^B(F+a)^{-1}\in S^0(\Xi,\X)\subset\mathfrak{C}^B(\Xi)$, $\mathfrak{z}^B_F(a)\in\mathfrak{C}^B
(\Xi)$, $\mathfrak{r}^B_F(\mathcal{z})\in\mathfrak{C}^B(\Xi)$ and we use the Theorem \ref{T-bound} for the composition of symbols and the fact that $\mathfrak{C}^B
(\Xi)$ is a $^*$-algebra for the magnetic Moyal product.
We conclude by using Proposition 6.29 in \cite{IMP2} that in fact $\mathfrak{r}^B_F(\mathcal{z})\in S^{-p}(\Xi,\X)$ for any $\mathcal{z}$ in the resolvent set of $\Op^A(F)$.

\smallskip
From \eqref{F-invers} we easily deduce that $\Im\text{\sf m}\,\mathfrak{r}^B_F(\mathcal{z})=\mathscr{H}^
p_A(\X)$. Applying $\Op^A\big(\mathfrak{r}^B_F(i)\big)$ to $\mathscr{S}(\X)$ and taking into account that $\mathfrak{r}^B_F(i)\in S^{-p}(\Xi,\X)\subset\mathfrak{M}^B(\Xi)$, that $\|\Op^A\big(\mathfrak{r}^B_F(i)\big)\|_{\mathbb{B}(L^2(\X))}\leq 1$ and that $\mathscr{S}(\X)$ is dense in $L^2(\X)$, we easily obtain the essential self-adjointness of $\Op^A(F)$ on $\mathscr{S}(\X)$.
\end{proof}

\subsection{The evolution group.}

We suppose given a magnetic field $B\in\Lbc$ and a real elliptic symbol $h\in S^p(\Xi,\X)$ for some $p>0$ and for some vector potential $A\in\Apol$ associated to $B$ we consider the self-adjoint operator $\Op^A(h):\mathscr{H}^p_A(\X)\rightarrow L^2(\X)$ as the one studied in the previous subsection, that we shall denote by $\mathfrak{Q}^A(h)$. Then, by Stone Theorem, we can consider the following one-parameter strongly continuous unitary group
$$
\mathbb{R}\ni t\mapsto W^A_h(t)\in\mathcal{U}\big(L^2(\X)\big).
$$
It is defined as the unique solution of the Cauchy problem
\beq
\left\{
\begin{array}{l}
i\partial_tW^A_h(t)\,=\,\mathfrak{Q}^A(h)W^A_h(t),\quad\forall t\in\mathbb{R}\\
W^A_h(0)\,=\,\bb1
\end{array}
\right.
\eeq
and given explicitly by the following formula (using the functional calculus with self-adjoint operators):
$$
W^A_h(t)\,=\,\exp\big(-it\mathfrak{Q}^A(h)\big).
$$

\begin{remark}\label{R-unit-gr}
For any $t\in\mathbb{R}$, the unitary operator $W^A_h(t)$ leaves invariant the domain $\mathscr{H}^p_A(\X)$ and we have the following commutation relation:
$$
W^A_h(t)\mathfrak{Q}^A(h)f\,=\,\mathfrak{Q}^A(h)W^A_h(t)f,\qquad\forall f\in\mathscr{H}^p_A(\X).
$$
\end{remark}

Let us consider its distribution symbol defined by 
$$
W^A_h(t)\,=:\,\Op^A\big(w^B_h(t)\big).
$$
A priori we know that $w^B_h(t)\in\mathfrak{C}^B(\Xi)$ for any $t\in\mathbb{R}$ and that it defines by magnetic quantization an invertible operator with the inverse having the following symbol (usually we denote by $F^-_B$ the inverse of $F\in\mathscr{S}^\prime(\Xi)$ for the magnetic Moyal product $\sharp^B$, when this inverse exists):
$$
\big[w^B_h(t)\big]^-_B=w^B_h(-t)=\overline{w^B_h(t)}
\in\mathfrak{C}^B(\Xi).
$$
We also know that the function $\mathbb{R}\ni t\mapsto w^B_h(t)\in\mathfrak{C}^B(\Xi)$ is a solution of the Cauchy problem
\beq
\left\{
\begin{array}{l}
i\partial_tw^B_h(t)\,=\,h\sharp^Bw^B_h(t),\quad\forall t\in\mathbb{R}\\w^B_h(0)\,=\,1.
\end{array}
\right.
\eeq

\begin{theorem}\label{T-ev-group}
Suppose given a magnetic field $B\in\Lbc$ and a real elliptic symbol $h\in S^p(\Xi,\X)$ for some $p>0$ and for some vector potential $A\in\Apol$ associted to $B$ let us consider the self-adjoint operator $\mathfrak{Q}^A(h):\mathscr{H}^p_A(\X)\rightarrow L^2(\X)$ and its associated unitary group $\big\{W^A_h(t)\big\}_{t\in\mathbb{R}}$. Then $W^A_h(t)\mathscr{S}(\X)\subset\mathscr{S}(\X)$ for any $t\in\mathbb{R}$.
\end{theorem}

\begin{proof}
From Remark \ref{R-unit-gr} and the definition pf $\mathfrak{Q}^A(h)$, we conclude that $W^A_h(t)\mathscr{S}(\X)\subset\mathscr{H}^p_A(\X)$. In order to prove the Theorem it is clearly enough to prove that 
$$
Q^\alpha\big(\Pi^A\big)^\beta W^A_h(t)\phi\in L^2(\X),\ \forall\phi\in\mathscr{S}(\X),\ \forall(\alpha,\beta)\in\mathbb{N}^{2d}.
$$
In dealing with these computations we shall use some notation.
\begin{description}
\item[Notation:] ~
\begin{itemize}
\item $q_j(x,\xi):=x_j$ and $p_j(x,\xi):=\xi_j$, for $1\leq j\leq d$.
\item Given some distribution $F\in\mathfrak{M}^B(\Xi)$ we set 
$$
\ad^B_{q_j}(F):=q_j\sharp^BF-F\sharp^Bq_j,\qquad\ad^B_{p_j}(F):=p_j\sharp^BF-F\sharp^Bp_j.
$$
\item For $1\leq j\leq d$ and for any $k\in\mathbb{N}$ we denote by $p_j^{\sharp k}$ the magnetic Moyal product of $k$ factors $p_j$ and similarly $p^{\sharp\gamma}:=(p_1^{\sharp\gamma_1})\sharp^B
\ldots\sharp^B(p_d^{\sharp\gamma_d})$ with multi-index notation. We also use similar notations for the symbols $\{q_1,\ldots,q_d\}$.
\item For any multi-index $\alpha\in\mathbb{N}^d$ we denote by $\{\alpha\}$ the ordered set with $\alpha_1$ entries equal to 1, followed by $\alpha_2$ entries equal to 2 and so on up to the last $\alpha_d$ entries equal to $d$. Reciprocally, for any subset $\mathfrak{m}\subset\{\gamma\}$ for some given $\gamma\in\mathbb{N}^d$ we denote by $\gamma_{\mathfrak{m}}\in\mathbb{N}^d$ its associated multi-index. 
\item For any $\gamma\in\mathbb{N}^d$ we shall denote by $[\ad^B_q]^\gamma(F)$ the multiple commutator
$$
[\ad^B_q]^\gamma(F)\,:=\,\ad^B_{q_{i_1}}\circ\ldots\circ\ad^B_{q_{i_{|\gamma|}}}\!(F),
$$
where $\{i_1,\ldots,i_{|\gamma|}\}=\{\gamma\}$, and similarly for $[\ad^B_p]^\gamma$.
\end{itemize}
\end{description}

\smallskip
We notice that all the distributions $q_1,\ldots,q_d$ and $p_1,\ldots,p_d$ are elements of $\mathfrak{M}^B(\Xi)$ (see \cite{MP1}) and the distributions $p_1,\ldots,p_d$ are also in $S^1(\Xi,\X)$. We shall use several times the following commutation formula
$$
\big(F_1\sharp^B\ldots\sharp^BF_N\big)\sharp^BG-G\sharp^B\big(F_1\sharp^B\ldots\sharp^BF_N\big)=
$$
\beq\label{F-mult-comm}
=\underset{1\leq k\leq N}{\sum}\underset{\{j_1,\ldots,j_k\}\subset\{1,\ldots,N\}}{\sum}[F_{j_1},\big[F_{j_2}\ldots[F_{j_k},G]_B\ldots\big]_B\Big]_B\sharp^B 
F_{l_1}\sharp^B\ldots\sharp^BF_{l_{N-k}}
\eeq
where $[F,G]_B:=F\sharp^BG-G\sharp^BF$ and for any subset $\{j_1,\ldots,j_k\}\subset\{1,\ldots,N\}$ we denote by $\{l_1,\ldots,l_{N-k}\}=\{1,\ldots,N\}\setminus\{j_1,\ldots,j_k\}$ all the sets being considered ordered by the natural order induced by $\mathbb{N}$.

\smallskip
For any pair $(\alpha,\beta)\in\mathbb{N}^{2d}$ we have that
$$
Q^\alpha\big(\Pi^A\big)^\beta=\Op^A\big(q^{\sharp\alpha}\sharp^Bp^{\sharp\beta}\big)\in\mathcal{L}\big(\mathscr{S}(\X);\mathscr{S}^\prime(\X)\big),
$$
so that for any test function $\phi\in\mathscr{S}(\X)$ the following tempered distribution is well defined:
$$
\big[Q^\alpha\big(\Pi^A\big)^\beta W^A_h(t)\phi\big](\psi)=\big[W^A_h(t)\phi\big]\big(\big(\Pi^A\big)^\beta Q^\alpha\psi)=\left\langle\overline{\big(\Pi^A
\big)^\beta Q^\alpha\psi}\,,\,W^A_h(t)\phi\right\rangle_{L^2(\X)}\hspace*{-0,7cm}
\in\mathbb{C},\quad\forall\psi\in\mathscr{S}(\X).
$$
We shall prove that it defines in fact a continuous functional of $\psi\in\mathscr{S}(\X)$ for the topology induced by $\mathscr{H}^p_A(\X)$. The idea is to compute the commutator
$$
\Big[Q^\alpha\big(\Pi^A\big)^\beta\,,\,W^A_h(t)\Big]=\Op^A\Big(\big(q^{\sharp\alpha}
\sharp^Bp^{\sharp\beta}\big)\sharp^Bw^B_h(t)-w^B_h(t)\sharp^B\big(q^{\sharp\alpha}\sharp^Bp^{\sharp\beta}\big)\Big)
$$
by computing its distribution kernel and having in mind the composition laws \eqref{Ext-sharp-B-2} and the fact  that the symbol $q^{\sharp\alpha}\sharp^Bp^{\sharp\beta}$ is in the magnetic Moyal algebra $\mathfrak{M}^B(\Xi)$. In order to deal with the commutators with $W^A_h(t)$ we notice that given some operator $X^A:=\Op^A(F_X)$ for some $F_X\in\mathfrak{M}^B(\Xi)$ we can write
\begin{align}
\big[X^A,\,W^A_h(t)\big]\phi&=\big(X^AW^A_h(t)-W^A_h(t)X^A\big)\phi=\int_0^tds\big[\partial_s\left(W^A_h(t-s)X^AW^A_h(s)\right)\big]\phi\nonumber\\
&=-i\int_0^tds\Big[W^A_h(t-s)\Big(X^A\mathfrak{Q}^A(h)-\mathfrak{Q}^A(h)X^A\Big)W^A_h(s)\Big]\phi\nonumber\\
&=-i\int_0^tdsW^A_h(t-s)\Big[X^A\,,\,\mathfrak{Q}^A(h)\Big]W^A_h(s)\phi.\nonumber
\end{align}
We are interested for the moment to study the case $X^A=\Op^A\big(q^{\sharp\alpha}\sharp^B p^{\sharp\beta}\big)\big)$ and the commutator
$$
\Big[X^A,\,\mathfrak{Q}^A(h)\Big]=\Op^A\big((q^{\sharp\alpha}\sharp^B p^{\sharp\beta})\sharp^Bh-h\sharp^B(q^{\sharp\alpha}\sharp^B p^{\sharp\beta})\big).
$$
We are going to proceed by induction on $N:=|\alpha+\beta|\in\mathbb{N}$, starting with the following \textit{induction hypothesis}:
\begin{description}
\item[H$_N$:] \emph{
Suppose that for any $(\alpha,\beta)\in\mathbb{N}^{2d}$ with $|\alpha|+|\beta|\leq N$ we know that $Q^\alpha\big(\Pi^A\big)^\beta\,W^A_h(t)\phi\in\mathscr{H}
^p_A(\X)$, for any $\phi\in\mathscr{S}(\X)$ and for any $t\in\mathbb{R}$.}
\end{description}
First of all let us notice that the above statement is clearly true for $N=0$ due to the fact that $\mathscr{S}(\X)\subset\mathscr{H}^p_A(\X)$ and the unitary group $W^A_h(t)$ leaves invariant the domain of its generator $\mathfrak{Q}^A(h):\mathscr{H}^p_A(\X)\rightarrow L^2(\X)$. 

\smallskip
Suppose now that we increase $N$ in the Hypothesis H$_N$ above by 1 either by increasing some $\alpha_j$ or some $\beta_k$ by 1 and let us consider some pair $(\alpha',\beta')\in\mathbb{N}^{2d}$ with $|\alpha'|+|\beta'|=N+1$.
Using formula \eqref{F-mult-comm} we can compute the following commutator in the magnetic Moyal algebra:
\beq\label{F-comm-h}
\Big[\Op^A\big(q^{\sharp\alpha'}\sharp^B p^{\sharp\beta'}\big)\big)\,,\,\mathfrak{Q}^A(h)\Big]=\underset{\mathfrak{m}\subset\{\alpha\}}{\sum}^{\prime}\,\underset{\mathfrak{n}\subset\{\beta\}}{\sum}^\prime\Op^A\Big([\ad^B_q]^{\gamma_{\mathfrak{m}}}\big([\ad^B_p]
^{\gamma_{\mathfrak{n}}}(h)\big)\Big)Q^{\gamma_{\mathfrak{m}^\complement}}\big(\Pi^A\big)^
{\gamma_{\mathfrak{n}^\complement}}=
\eeq
$$
=\underset{\mathfrak{m}\subset\{\alpha\}}{\sum}^\prime\,\underset{\mathfrak{n}\subset\{\beta\}}{\sum}^\prime\Op^A\Big([\ad^B_q]^{\gamma_{\mathfrak{m}}}\big([\ad^B_p]
^{\gamma_{\mathfrak{n}}}(h)\big)\Big)\,\mathfrak{R}^B_-\big(\mathfrak{Q}(h)+i\big)\,Q^{\gamma_{\mathfrak{m}^\complement}}\big(\Pi^A\big)^{\gamma_{\mathfrak{n}^\complement}},
$$
where $\mathfrak{m}^\complement=\{\alpha\}\setminus
\mathfrak{m}$ and $\mathfrak{n}^\complement=\{\beta\}\setminus
\mathfrak{n}$ and we denote by $\underset{\mathfrak{m}\subset\{\alpha\}}{\sum'}$ the sum over all non-trivial subsets (i.e. different from the void set and from the total set), so that $|\gamma_{\mathfrak{m}^\complement}|+|\gamma_{\mathfrak{n}^\complement}|\leq N$.

\smallskip
Finally, for any $\phi\in\mathscr{S}(\X)$ we have obtained the following equality of tempered distributions on $\X$:
$$
Q^{\alpha'}\big(\Pi^A\big)^{\beta'}W^A_h(t)\phi=W^A_h(t)Q^{\alpha'}\big(\Pi^A\big)^{\beta'}\phi\ -
$$
$$
-i\underset{\mathfrak{m}\subset\{\alpha\}}{\sum}^\prime\,\underset{\mathfrak{n}\subset\{\beta\}}{\sum}^\prime\int_0^tdsW^A_h(t-s)\Op^A\Big([\ad^B_q]^{\gamma_{\mathfrak{m}}}\big([\ad^B_p]
^{\gamma_{\mathfrak{n}}}(h)\big)\sharp^B\mathfrak{r}^B_-\Big)\big(\mathfrak{Q}(h)+i\big)\,Q^{\gamma_{\mathfrak{m}^\complement}}\big(\Pi^A\big)^
{\gamma_{\mathfrak{n}^\complement}}W^A_h(s)\phi.
$$
We notice that
$$
Q^{\alpha'}\big(\Pi^A\big)^{\beta'}\phi\in\mathscr{S}(\X)\ \Longrightarrow\ 
W^A_h(t)Q^{\alpha'}\big(\Pi^A\big)^{\beta'}\phi\in
\mathscr{H}^p_A(\X),
$$
$$
\text{\bf H}_N\ \Longrightarrow\ \big(\mathfrak{Q}(h)+i\big)\,Q^{\gamma_{\mathfrak{m}^\complement}}\big(\Pi^A\big)^
{\gamma_{\mathfrak{n}^\complement}}W^A_h(s)\phi\in L^2(\X),\quad\forall s\in[0,t],
$$
for any subsets $\mathfrak{m}\subset\{\alpha\}$ and $\mathfrak{n}\subset\{\beta\}$ different from the empty set and of the total set. Proposition \ref{P-2-nd-symb-comp} implies
$$
[\ad^B_p]^{\gamma_{\mathfrak{n}}}(h)\in S^p(\Xi,\X),\quad\forall\mathfrak{n}\subset\{\beta\},\emptyset\ne\mathfrak{n}\ne\{\beta\},
$$
$$
[\ad^B_q]^{\gamma_{\mathfrak{m}}}\big([\ad^B_p]
^{\gamma_{\mathfrak{n}}}(h)\big)\in S^{p-|\gamma_{\mathfrak{n}}|}(\Xi,\X),\quad\forall\mathfrak{n}\subset\{\beta\},\emptyset\ne\mathfrak{n}\ne\{\beta\},\ \forall\mathfrak{m}\subset\{\alpha\},\emptyset\ne\mathfrak{m}\ne\{\alpha\}
$$
and we conclude that 
$$
[\ad^B_q]^{\gamma_{\mathfrak{m}}}\big([\ad^B_p]
^{\gamma_{\mathfrak{n}}}(h)\big)\sharp^B\mathfrak{r}^B_-\in S^0(\Xi,\X)\subset\mathfrak{C}^B(\Xi)
$$
and finally
$$
Q^{\alpha'}\big(\Pi^A\big)^{\beta'}W^A_h(t)\phi\in L^2(\X).
$$

We consider now the following equality of tempered distributions:
$$
\mathfrak{Q}^A(h)Q^{\alpha'}\big(\Pi^A\big)^{\beta'}W^A_h(t)\phi\ =\ Q^{\alpha'}\big(\Pi^A\big)^{\beta'}\mathfrak{Q}^A(h)
W^A_h(t)\phi\,+\,\Big[\mathfrak{Q}^A(h)\,,\,Q^{\alpha'}\big(\Pi^A\big)^{\beta'}\Big]W^A_h(t)\phi\ =
$$
$$
Q^{\alpha'}\big(\Pi^A\big)^{\beta'}
W^A_h(t)\mathfrak{Q}^A(h)\phi\,+\,\Big[\mathfrak{Q}^A(h)\,,\,Q^{\alpha'}\big(\Pi^A\big)^{\beta'}\Big]W^A_h(t)\phi
$$
and using the above result and once again formula \eqref{F-comm-h} we conclude that it defines in fact an element in $L^2(\X)$. This proves that \textbf{H}$_{(N+1)}$ is also true and finishes the proof of the Theorem.
\end{proof}

\section{Appendices.}

\subsection{A.1: Estimations on the derivatives of $\varomega^B$.}

We shall consider the multi-indices $\{\sigma^j\}_{1\leq j\leq d}$ with $(\sigma^j)_k:=\delta_{jk}$. We shall use the notation $\varomega^B_x(y,z):=e^{-iF^B_x(y,z)}$ with the explicit expression:
$$
F^B_x(y,z):=\int_{\mathcal{T}_x(y,z)}\hspace*{-12pt}B=4\underset{j\ne k}{\sum}y_jz_k\int_0^1ds\int_0^sdt\,B_{jk}
\big(x+(2s-1)y+(2t-1)z\big)=:4\underset{j\ne k}{\sum}y_jz_k\Omega^B_{jk}(x,y,z).
$$
Then we have the following formulas:
$$
\big(\partial^\alpha_xF^B_x\big)(y,z)=4\underset{j\ne k}{\sum}y_jz_k\int_0^1ds\int_0^sdt\big(\partial^\alpha B_{jk}\big)
\big(x+(2s-1)y+(2t-1)z\big)=\int_{\mathcal{T}_x(y,z)}\hspace*{-10pt}\partial^\alpha B=F^{\partial^\alpha B}_x(y,z),
$$
\beq
\begin{split}
\big(\partial^\alpha_yF^B_x\big)(y,z)&=4\underset{j\ne k}{\sum}y_jz_k\int_0^1ds\int_0^sdt\big(\partial^\alpha B_{jk}\big)\big(x+(2s-1)y+(2t-1)z\big)(2s-1)^{|\alpha|}\,+\\
&+\underset{j:\alpha_j\geq1}{\sum}\ \underset{k\ne j}{\sum}z_k\int_0^1ds\int_0^sdt\big(\partial^{\alpha-\sigma^j} B_{jk}\big)
\big(x+(2s-1)y+(2t-1)z\big)(2s-1)^{|\alpha|-1},
\end{split}
\eeq
\beq
\begin{split}
\big(\partial^\alpha_zF^B_x\big)(y,z)&=4\underset{j\ne k}{\sum}y_jz_k\int_0^1ds\int_0^sdt\big(\partial^\alpha B_{jk}\big)
\big(x+(2s-1)y+(2t-1)z\big)(2t-1)^{|\alpha|}\,+\\
&+\underset{k:\alpha_k\geq1}{\sum}\ \underset{j\ne k}{\sum}y_j\int_0^1ds\int_0^sdt\big(\partial^{\alpha-\sigma^k} B_{jk}\big)
\big(x+(2s-1)y+(2t-1)z\big)(2t-1)^{|\alpha|-1}.
\end{split}
\eeq

Now let us recall the Fa\`{a} di Bruno's formula (\cite{G}) for the case of the exponential of a given function $F^B$:
\beq
\begin{split}
\partial^\alpha_x\varomega^B_x(y,z)&=\alpha!\left(\underset{1\leq l\leq|\alpha|}{\sum}\frac{1}{l!}\ \sum\limits_{\gamma^1+\cdots+\gamma^l=\alpha}
\prod\limits_{s=1}^{l}\frac{1}{\gamma^s!}\big(\partial^{\gamma^s}_xF^B_x\big)(y,z)\right)\varomega^B_x(y,z)\\
&=\alpha!\left(\underset{1\leq l\leq|\alpha|}{\sum}\frac{1}{l!}\ \sum\limits_{\gamma^1+\cdots+\gamma^l=\alpha}
\prod\limits_{s=1}^{l}\frac{1}{\gamma^s!}\left(\int_{\mathcal{T}_x(y,z)}\hspace*{-10pt}\partial^{\gamma^s} B\right)\right)\varomega^B_x(y,z).
\end{split}
\eeq
It follows then
\beq
\begin{split}
\left|\partial^\alpha_x\varomega^B_x(y,z)\right|&\leq C(d,|\alpha|)\underset{1\leq l\leq|\alpha|}{\max}\ \underset{p_1+\cdots+p_l=|\alpha|}{\max}\ \prod\limits_{s=1}^{l}\big(|y\wedge z|^{p_s}\rho_{p_s}(B)\big)\\
&\leq C(d,|\alpha|)|y\wedge z|^{|\alpha|}\underset{1\leq l\leq|\alpha|}{\max}\ \underset{p_1+\cdots+p_l=|\alpha|}{\max}\ \prod\limits_{s=1}^{l}\mu_{p_s}(B).
\end{split}
\eeq
\beq
\begin{split}
\partial^\alpha_y\varomega^B_x(y,z)&=\alpha!\left(\underset{1\leq l\leq|\alpha|}{\sum}\frac{1}{l!}\ \sum\limits_{\gamma^1+\cdots+\gamma^l=\alpha}
\prod\limits_{s=1}^{l}\frac{1}{\gamma^s!}\big(\partial^{\gamma^s}_yF^B\big)(x,y,z)\right)\varomega^B_x(y,z)\\
&=\alpha!\left(\underset{1\leq l\leq|\alpha|}{\sum}\frac{1}{l!}\ \sum\limits_{\gamma^1+\cdots+\gamma^l=\alpha}
\prod\limits_{s=1}^{l}\frac{1}{\gamma^s!}\left(\int_{\mathcal{T}_x(y,z)}
\hspace*{-10pt}\partial^{\gamma^s} B+ \underset{j:\alpha_j\geq1}{\sum}\mathfrak{S}^{|\gamma^s|,1}_{\mathcal{T}_x(y,z)}\big[\big(\partial^{\gamma^s-
\sigma^j}B\big)\llcorner z\big]_j\right)\right)\varomega^B_x(y,z),
\end{split}
\eeq
$$
\left|\partial^\alpha_y\varomega^B_x(y,z)\right|\leq C(d,|\alpha|)<y>^{|\alpha|}<z>^{|\alpha|}\underset{1\leq l\leq|\alpha|}{\max}\ \underset{p_1+\cdots+p_l=|\alpha|}{\max}\ \prod\limits_{s=1}^{l}(\mu_{p_s}(B)+\mu_{p_s-1}(B)).
$$
\beq
\begin{split}
\partial^\alpha_z\varomega^B_x(y,z)&=\alpha!\left(\underset{1\leq l\leq|\alpha|}{\sum}\frac{1}{l!}\ \sum\limits_{\gamma^1+\cdots+\gamma^l=\alpha}
\prod\limits_{s=1}^{l}\frac{1}{\gamma^s!}\big(\partial^{\gamma^s}_zF^B\big)(x,y,z)\right)\varomega^B_x(y,z)\\
&=\alpha!\left(\underset{1\leq l\leq|\alpha|}{\sum}\frac{1}{l!}\ \sum\limits_{\gamma^1+\cdots+\gamma^l=\alpha}
\prod\limits_{s=1}^{l}\frac{1}{\gamma^s!}\left(\int_{\mathcal{T}_x(y,z)}
\hspace*{-10pt}\partial^{\gamma^s} B+\underset{k:\alpha_k\geq1}{\sum}\mathfrak{S}^{|\gamma^s|,2}_{\mathcal{T}_x(y,z)}\big[y\lrcorner\big(
\partial^{\gamma^s-\sigma^k}B\big)\big]_k\right)\right)\varomega^B_x(y,z),
\end{split}
\eeq
$$
\left|\partial^\alpha_z\varomega^B_x(y,z)\right|\leq C(d,|\alpha|)<y>^{|\alpha|}<z>^{|\alpha|}\underset{1\leq l\leq|\alpha|}{\max}\ \underset{p_1+\cdots+p_l=|\alpha|}{\max}\ \prod\limits_{s=1}^{l}(\mu_{p_s}(B)+\mu_{p_s-1}(B)).
$$

\subsection{A.2: Estimating an oscillating integral.}

In the following technical statement we shall need some specific weights on the space $C^\infty_{\text{\sf pol}}\big(\X\times\X\times\X\big)$. For any family $(N,M,n_1,n_2,m_1,m_2)\in\mathbb{N}^6$ we define for any $\Theta\in C^\infty_{\text{\sf pol}}\big(\X\times\X\times\X\big)$:
\beq\label{FD-seminorms-3-1}
\mathcal{W}^{N,n_1,n_2}_{M,m_1,m_2}(\Theta):=\!\underset{(x,y,z)\in\X\times\X\times\X}{\sup}\left<x\right>^{-N}\!\left<y\right>^{-n_1}\!\left<z\right>^{-n_2}\!\underset{|a|\leq M}{\max}\underset{|b|\leq m_1}{\max}\underset{|c|\leq m_2}{\max}\left|\partial^a_x\partial^b_y\partial^c_z\Theta(x,y,z)\right|
\eeq
that can take also the value $+\infty$.

\smallskip
\begin{proposition}\label{P-Ap-1}
Suppose given a magnetic field of class $\Lbc$ and two H\"{o}rmander type symbols $(F,G)\in S^{p_1}(\Xi,\X)\times S^{p_2}(\Xi,\X)$ and let us denote by 
$n_1:=2[(d+p_1)/2]+2$, $n_2:=2[(d+p_2)/2]+2$. 
Suppose also given a function $\Theta\in C^\infty_{\text{\sf pol}}\big(\X;C^\infty_{\text{\sf pol}}(\X\times\X)\big)$
such that  $\mathcal{W}^{N,q_1,q_2}_{n,n_2,n_1}(\Theta)<\infty$ for some $(N,q_1,q_2,n)\in\mathbb{R}_+^3\times\mathbb{N}$.
Then, for any $(\alpha,\beta)\in\mathbb{N}^2d$ with $|\alpha|\leq n$ we have the estimations:
$$
\left<x\right>^{-N}\left<\xi\right>^{-(p_1+p_2)+|\beta|}\left|\partial^\alpha_x\partial^\beta_\xi\int_{\Xi}dY\int_{\Xi}dZ
e^{-2i(<\eta,z>-<\zeta,y>)}\varomega^B_x(y,z)\Theta(x,y,z)F(X-sY)G(X-tZ)\right|\,\leq
$$
$$
\leq C(d,p_1,p_2,\alpha,\beta)\mathfrak{w}_{|\alpha|+n_1+n_2}(B)\,\mathcal{W}^{N,q_1,q_2}_{|\alpha|,n_2,n_1}(\Theta)\underset{0\leq k\leq|\beta|}{\sum}\rho^{p_1-k}_{|\alpha|+n_2;k,m_2}(F)\,\rho^{p_2-|\beta|+k}_{|\alpha|+n_1;|\beta|-k,m_1}(G),
$$
where $m_1:=2[n_2+(n+n_1+q_1)/2]+2$ and $m_2:=2[n_1+(n+n_2+q_2)]+2$.
\end{proposition}

\begin{proof}
Fixing some $(x,\xi)\in\Xi$, we have to estimate oscillating integrals of the form
\beq\label{F-sharp-B-symbols-2}
\left<\xi\right>^{-(p_1+p_2)+|\beta|}\int_{\Xi}dY\int_{\Xi}dZ\,e^{-2i(<\eta,z>-<\zeta,y>)}\left[\partial_x^\alpha\partial_\xi^\beta\big(
\varomega^B_x(y,z)\,\Theta(x,y,z)\,F(X-Y)\,G(X-Z)\big)\right],
\eeq
that can be written by the Leibnitz rule as finite linear combinations depending only on $|\alpha|\in\mathbb{N}$ and $|\beta|\in\mathbb{N}$ of terms of the form
$$
\left<\xi\right>^{-(p_1+p_2)+|\beta|}\int_{\Xi}dY\int_{\Xi}dZ\,e^{-2i(<\eta,z>-<\zeta,y>)}\left[\big(\partial_x^{\alpha_0}\varomega^B_x\big)(y,z)\right]\left[\partial_x^{a_0}\Theta(x,y,z)\right]\,\times
$$
$$
\times\,\left[\big(\partial_x^{\alpha_1}\partial_\xi^{\beta_1}F\big)(X-Y)\right]\left[\big(\partial_x^{\alpha_2}\partial_\xi^{\beta_2}G\big)(X-Z)\right].
$$
Let us begin with a rough estimation of the 'momentum integrals' with respect to $(\eta,\zeta)\in\big[\X^*\big]^2$:
$$
\left<\xi\right>^{-(p_1+p_2)+|\beta|}\left[\big(\partial_x^{\alpha_1}\partial_\xi^{\beta_1}F\big)(X-Y)\right]\left[\big(\partial_x^{\alpha_2}\partial_\xi^{\beta_2}G\big)(X-Z)\right]\leq
$$
$$
\leq C\nu^{p_1-|\beta|_1}_{|\alpha_1||\beta_1|}(F)\,\nu^{p_2-|\beta_2|}_{|\alpha_2||\beta_2|}(G)\leq C\nu^{p_1}_{|\alpha||\beta|}(F)\,\nu^{p_2}_{|\alpha||\beta|}(G),
$$
with $\alpha_0+\alpha_1+\alpha_2=\alpha-a_0$ and $\beta_1+\beta_2=\beta$. Moreover, in order to control these integrals, some extra factors of convergence $<\eta>^{-d-\epsilon}<\zeta>^{-d-\epsilon}$ have to be introduced. We are then obliged to get rid of these growing factors, integrating by parts using the identities
$$
\left<\eta\right>^{n_1}\!e^{-2i<\eta,z>}=\left<(i/2)\nabla_z\right>^{n_1}\!e^{-2i<\eta,z>};\quad\left<\zeta\right>^{n_2}e^{2i<\zeta,y>}=\left<(1/2i)\nabla_y\right>^{n_2}\!e^{2i<\zeta,y>},
$$
taking $n_1=2[(d+p_1)/2]+2$ and $n_2=2[(d+p_2)/2]+2$ (in order to work with polynomials in the differential operators we have to take even exponents). Thus the integral \eqref{F-sharp-B-symbols-2} becomes a linear combination (depending only on $\{p_1,p_2,|\alpha|,|\beta|,d\}$) of terms of the form
$$
\left<\xi\right>^{-(p_1+p_2)+|\beta|}\int_{\Xi}dY\int_{\Xi}dZ\,e^{-2i(<\eta,z>-<\zeta,y>)}\left[\big(\partial_x^{\alpha_0}\partial_y^{\mu_0}
\partial_z^{\nu_0}\varomega^B_x\big)(y,z)\right]\left[\partial_x^{a_0}\partial_y^{b_0}\partial_z^{c_0}\Theta(x,y,z)\right]\,\times
$$
$$
\times\,\left<\eta\right>^{-n_1}\left[\big(\partial_x^{\alpha_1}\partial_y^{\mu_1}\partial_\xi^{\beta_1}F\big)(X-Y)\right]\,\left<\zeta\right>^{-n_2}
\left[\big(\partial_x^{\alpha_2}\partial_z^{\nu_1}\partial_\xi^{\beta_2}G\big)(X-Z)\right],
$$
where $a_0+\alpha_0+\alpha_1+\alpha_2=\alpha$, $\beta_1+\beta_2=\beta$, $|b_0+\mu_0+\mu_1|=n_2$ and $|c_0+\nu_0+\nu_1|=n_1$.
A maximum number of $N_0:=|\alpha|+n_1+n_2$ derivatives of the factor $\varomega^B_x$ will appear. Considering the factor $\partial_x^{a_0}\partial_y^{b_0}
\partial_z^{c_0}\Theta(x,y,z)$ in the oscillating integral above we notice that $|a_0|\leq|\alpha|$, $|b_0|\leq n_2$ and $|c_0|\leq n_1$.

\smallskip
We notice that for any function $\Theta\in C^\infty_{\text{\sf pol}}\big(\X;C^\infty_{\text{\sf pol}}(\X\times\X)\big)$ there exist some functions $w_y:\mathbb{N}\times\mathbb{N}\rightarrow\mathbb{N}$, $w_z2:\mathbb{N}\times\mathbb{N}\rightarrow\mathbb{N}$ and $w_x:\mathbb{N}\times\mathbb{N}\times\mathbb{N}\rightarrow
\mathbb{N}$ depending on $\Theta\in C^\infty_{\text{\sf pol}}\big(\X;C^\infty_{\text{\sf pol}}(\X\times\X)\big)$ such that if we set
\beq\label{FD-seminorms-3-2}
\mathcal{W}^{w_y(m_1,m_2),w_z(m_1,m_2)}_{m_1,m_2}(\Theta,a;x):=\hspace*{-8pt}\underset{(y,z)\in\X\times\X}{\sup}\hspace*{-8pt}\left<y\right>^{-n_1}\!\left<z\right>^{-n_2}\hspace*{-4pt}\underset{|b|\leq m_1,|c|\leq m_2}{\max}\left|\partial^a_x\partial^b_y\partial^c_z\Theta(x,y,z)\right|<\infty,\ \forall x\in\X,
\eeq
we can write
\beq\label{FD-seminorms-3-3}
\mathcal{W}^{w_x(M,m_1,m_2),w_y(m_1,m_2),w_z
(m_1,m_2)}_{M,m_1,m_2}(\Theta)=\underset{x\in\X}{\sup}\left<x\right>^{-N}\underset{|a|\leq M}{\max}\,\mathcal{W}^{w_y(m_1,m_2),w_z(m_1,m_2)}_{m_1,m_2}(\Theta,a;x)<\infty.
\eeq
Thus, with the notations \eqref{FD-seminorms-3-1}, \eqref{FD-seminorms-3-2}, \eqref{FD-seminorms-3-3} we have to choose 
\beq
\begin{split}
N=w_x(|\alpha|,n_2,n_1)\in\mathbb{N},\\
q_1=w_y(n_2,n_1),\\
q_2=w_z(n_2,n_1)
\end{split}
\eeq
associated to $\Theta\in C^\infty_{\text{\sf pol}}\big(\X;C^\infty_{\text{\sf pol}}(\X\times\X)\big)$.

\smallskip
In order to obtain integrability in the variables $(y,z)\in\X^2$, we shall insert the factors $\left<y\right>^{-m_1}\!\left<z\right>^{-m_2}$ with $m_1=2[(N_0+n_2+q_1)/2]+2$, $m_2=2[(N_0+n_1+q_2)/2]+2$ and apply once again integration by parts to transform the compensating factors in derivations with respect to $(\eta,\zeta)\in(\X^*)^2$. Finally we obtain a linear combination (depending on $\{p_1,p_2,|\alpha|,|\beta|,d\}$) of terms of the form
\beq\label{gen-term}
\left<\xi\right>^{-(p_1+p_2)+|\beta|}\int_0^1d\tau\int_{\Xi}\left<y\right>^{-m_1}\!\left<\eta\right>^{-n_1}dY\int_{\Xi}\left<z\right>^{-m_2}\!\left<\zeta\right>^{-n_2}dZ\,e^{-2i(<\eta,z>-<\zeta,y>)}\varomega^B_x(y,z)\,\times
\eeq
$$
\times\,\left[\big(\partial_x^{\alpha_0}\partial_y^{\mu_0}\partial_z^{\nu_0}\varomega^B_x\big)(y,z)\right]\,\left[\partial_x^{a_0}\partial_y^{b_0}\partial_z^{c_0}\Theta(x,y,z)\right]\,
\left[\big(\partial_x^{\alpha_1}\partial_y^{\mu_1}\partial_\xi^{\beta_1}\partial_\eta^{\theta_1}F\big)(X-Y)\right]\,
\left[\big(\partial_x^{\alpha_2}\partial_z^{\nu_1}\partial_\xi^{\beta_2}\partial_\zeta^{\theta_2}G\big)(X-Z)\right],
$$
where
\beq
\left\{
\begin{array}{l}
n_1=2[(d+p_1)/2]+2,\quad n_2=2[(d+p_2)/2]+2,\\
N_0:=|\alpha|+n_1+n_2,\\
m_1=2[(N_0+n_2)/2]+2,\quad m_2=2[(N_0+n_1)/2]+2,\\
|\gamma_1+\gamma_2+\sum\limits_{j=1}^{k}\lambda_j|\leq N_0,\\
0\leq n\leq N_0,\\
\gamma_1\leq\nu_0,\quad\gamma_2\leq\mu_0,\quad|b_0+\mu_0+\mu_1|=n_2,\quad|c_0+\nu_0+\nu_1|=n_1,\\
\alpha_1+\alpha_2\leq\alpha,\quad\beta_1+\beta_2=\beta,\\
|\theta_1|\leq m_2,\quad|\theta_2|\leq m_1.
\end{array}
\right.
\eeq
Due to the above remarks we can estimate each integral of the form \eqref{gen-term} by
$$
C<x>^N\mathfrak{w}_{|\alpha|+n_1+n_2}(B)\,\mathcal{W}^{N,q_1,q_2}_{|\alpha|,n_2,n_1}(\Theta)\underset{0\leq k\leq|\beta|}{\sum}\rho^{p_1-k}_{|\alpha|+n_2;k,m_2}(F)\,\rho^{p_2-|\beta|+k}_{|\alpha|+n_1;|\beta|-k,m_1}\!(G)
$$
and finish the proof of the Proposition.
\end{proof}


\smallskip
{\bf Acklowdlegements.} The authors have been supported by the Fondecyt Project 1160359.


\begin{thebibliography}{999}

\bibitem{ABG} W.O. Amrein, A. Boutet de Monvel, V. Georgescu: \textit{$C_0$-groups, commutator methods and spectral theory of N-body Hamiltonians}. Progress in Mathematics, vol. 135, Birkh\"{a}user Verlag, 1996 (p. xiv + 460).

\bibitem{CN} H. D. Cornean, G. Nenciu: \textit{Two dimensional magnetic Schrodinger operators: width of mini bands in the tight binding approximation}. Ann. Henri Poincar\'{e} 1, (2000), p. 203--222.

\bibitem{dA} Gianfausto Dell'Antonio:
\textit{Lectures on the Mathematics of Quantum Mechanics I}. Atlantis Studies in Mathematical Physics: Theory and Applications, Vol. 1. Atlantis Press 2015.

\bibitem{G} Henryk Gzyl: "Multidimensional extension of Fa\`{a} di Bruno's Formula". Journal of Mathematical Analysis and Applications 116, p. 450 -- 455, (1986).

\bibitem{H-I} L. H\"{o}rmander: \textit{The Analysis of Linear Partial Differential Operators I}. Springer-Verlag (2-nd ed.) 1990.

\bibitem{IMP1} Viorel Iftimie; Mantoiu, Marius; Purice, Radu: \textit{Magnetic pseudodifferential operators}, Publications of RIMS, 43 (2007), no. 3, 585â€--623.

\bibitem{IMP2} Viorel Iftimie, Marius Mantoiu, Radu Purice: \textit{Commutator Criteria for Magnetic Pseudodifferential Operators}. Comm. Partial Diff. Eq. 35 (2010), 1058--1094.

\bibitem{IMP3} Viorel Iftimie, Marius Mantoiu, Radu Purice: \textit{The magnetic formalism; new results}. Contemporary Mathematics 500 (2009), American Mathematical Society, p. 123--138.

\bibitem{MP1} Mantoiu, Marius; Purice, Radu: \textit{The magnetic Weyl calculus}. J. Math. Phys. 45 (2004), no. 4, 1394--1417.

\bibitem{MP2} Mantoiu, Marius; Purice, Radu: \textit{Strict deformation quantization for a particle in a magnetic field}. J. Math. Phys. 46 (2005), no. 5, 052105, 15 pp.

\bibitem{MP3} Mantoiu, Marius; Purice, Radu: \textit{The mathematical formalism of a particle in a magnetic field.} Mathematical physics of quantum mechanics, 417--434, Lecture Notes in Phys., 690, Springer, Berlin, 2006.

\bibitem{MPR} Mantoiu, Marius; Purice, Radu; Richard, Serge: \textit{Spectral and propagation results for magnetic Schodinger operators; a C*-Algebraic framework}, Journal of
Functional Analysis, 250 (2007), 42--67.

\bibitem{N-02} G. Nenciu
\textit{On asymptotic perturbation theory for Quantum Mechanics: Almost invariant subspaces and gauge invariant magnetic perturbation theory.} J. Math Phys. 43 (2002), p. 1273-1298.

\bibitem{RS1} Michael Reed, Barry Simon: \textit{Functional Analysis}. Academic Press, 1981.

\bibitem{Schw} Laurent Schwartz: \textit{Th\'{e}orie des Distributions}. Hermann, 1978.

\bibitem{T} Tr\'{e}ves, F.: \textit{Topological Vector Spaces, Distributions and Kernels}, Academic Press, New York, London, 1967.

\end{thebibliography}
\end{document}